\documentclass[a4paper,UKenglish,cleveref,autoref]{lipics-v2021}
\usepackage[T1]{fontenc}
\usepackage{amsmath,amssymb,stmaryrd}
\usepackage{tikz-cd}
\usepackage{stmaryrd}[only,leftrightarroweq,llbracket,rrbracket]
\usepackage[capitalise]{cleveref}
\usepackage{enumitem}
\usepackage{proof}
\usepackage{stackrel}

\usepackage[appendix=append,bibliography=common]{apxproof}

\newtheoremrep{theorem}{Theorem}
\newtheoremrep{lemma}[theorem]{Lemma}
\newtheoremrep{proposition}[theorem]{Proposition}
\newtheoremrep{corollary}[theorem]{Corollary}
\theoremstyle{definition}
\newtheoremrep{definition}[theorem]{Definition}
\newtheoremrep{example}[theorem]{Example}
\newtheoremrep{remark}[theorem]{Remark}

\usepackage[colorinlistoftodos,textsize=tiny,color=orange!70
, disable%
]{todonotes}

\newcommand{\mystrut}{\raisebox{-6pt}{}}

\newcommand{\power}{\mathcal P} %
\newcommand{\dist}{\mathcal D} %

\newcommand{\V}{\mathcal{V}}

\newcommand{\ev}{\textit{ev}}

\makeatletter
\newcommand{\dotminus}{\mathbin{\text{\@dotminus}}}

\newcommand{\@dotminus}{%
  \ooalign{\hidewidth\raise1ex\hbox{.}\hidewidth\cr$\m@th-$\cr}%
}
\makeatother

\newcommand{\pmet}{\mathbf{PMet}}
\newcommand{\dpmet}{\mathbf{DPMet}}

\newcommand{\Ccat}{\mathbf{C}}
\newcommand{\alg}[1]{\mathbf{EM}(#1)}
\newcommand{\pos}{\mathbf{Pos}}
\newcommand{\set}{\mathbf{Set}}

\DeclareMathOperator{\cl}{\mathsf{cl}}

\makeatletter
\newcommand{\eqnum}{\refstepcounter{equation}\textup{\tagform@{\theequation}}}
\makeatother

\DeclareMathOperator{\lo}{\mathsf{log}}
\DeclareMathOperator{\be}{\mathsf{beh}}

\newcommand{\sem}[1]{\llbracket #1 \rrbracket}

\newcommand{\tobject}{\Omega}

\newcommand{\elements}[1]{\int #1}
\DeclareMathOperator{\kant}{\kappa}
\newcommand{\EMlog}{\mathcal L_{\textrm{EM}}}

\newcommand{\qlog}{\mathcal L_{\Lambda}}

\title{Expressive Quantale-Valued Logics for Coalgebras: an
  Adjunction-Based Approach}

\ArticleNo{XX}

\nolinenumbers

\author{Harsh Beohar}{University of Sheffield, United Kingdom}{h.beohar@sheffield.ac.uk}{https://orcid.org/0000-0001-5256-1334}{}

\author{Sebastian Gurke}{Universität Duisburg-Essen, Germany}{sebastian.gurke@uni-due.de}{https://orcid.org/0009-0008-4343-1384}{}

\author{Barbara K{\"o}nig}{Universität Duisburg-Essen, Germany}{barbara\_koenig@uni-due.de}{https://orcid.org/0000-0002-4193-2889}{}

\author{Karla Messing}{Universität Duisburg-Essen, Germany}{karla.messing@uni-due.de}{https://orcid.org/0009-0003-1019-6449}{}

\author{Jonas Forster}{Friedrich-Alexander-Universität Erlangen-Nürnberg,
  Germany}{jonas.forster@fau.de}{https://orcid.org/0000-0002-5050-2565}{}

\author{Lutz Schr{\"o}der}{Friedrich-Alexander-Universität Erlangen-Nürnberg,
  Germany}{lutz.schroeder@fau.de}{https://orcid.org/0000-0002-3146-5906}{}

\author{Paul Wild}{Friedrich-Alexander-Universität Erlangen-Nürnberg,
  Germany}{paul.wild@fau.de}{https://orcid.org/0000-0001-9796-9675}{}

\authorrunning{H. Beohar, S. Gurke, B. K{\"o}nig, K. Messing,
  J. Forster, L. Schr{\"o}der, P. Wild}

\Copyright{H. Beohar, S. Gurke, B. K{\"o}nig, K. Messing,
  J. Forster, L. Schr{\"o}der, P. Wild}

\funding{The first author was supported by the EPSRC NIA Grant
  EP/X019373/1, while the remaining authors were supported by the
  Deutsche Forschungsgemeinschaft (DFG, German Research Foundation) --
  project number 434050016 (SpeQt).}


\begin{document}

\maketitle              %

\begin{abstract}
  We address the task of deriving fixpoint equations from modal logics
  characterizing behavioural equivalences and metrics (summarized
  under the term \emph{conformances}). We rely on an earlier work that
  obtains Hennessy-Milner theorems as corollaries to a fixpoint
  preservation property along Galois connections between suitable
  lattices. We instantiate this to the setting of coalgebras, in which
  we spell out the compatibility property ensuring that we can derive
  a behaviour function whose greatest fixpoint coincides with the
  logical conformance. We then concentrate on the linear-time case,
  for which we study coalgebras based on the machine functor living in
  Eilenberg-Moore categories, a scenario for which we obtain a
  particularly simple logic and fixpoint equation. The theory is
  instantiated to concrete examples, both in the branching-time
  case (bisimilarity and behavioural metrics) and in the linear-time
  case (trace equivalences and trace distances).

\end{abstract}

\section{Introduction}

Behavioural equivalences (such as bisimilarity and trace equivalence)
are an important technique to identify states with the same behaviour
in a transition system~\cite{g:linear-branching-time}.  They have been
complemented by notions of behavioural metrics
\cite{GiacaloneEA90,bw:behavioural-pseudometric,dgjp:metrics-labelled-markov}
measuring the distance between states, in particular in a quantitative
setting. We work in a coalgebraic setting~\cite{Rutten00} that allows
us to answer generic questions about behavioural equivalences and
metrics, parametrized over various branching types (non-deterministic,
probabilistic, weighted, etc.).

There are various ways to characterize behavioural equivalences or
metrics, which we illustrate using trace equivalence as an example:
(i) \emph{direct specification:} Two states $x,y$ are trace equivalent
if they admit the same traces; (ii) \emph{logic:} $x,y$ are trace
equivalent if they cannot be distinguished in a modal logic based on
diamond modalities and the constant \emph{true}; (iii)
\emph{fixpoint equation:} $x,y$ are trace equivalent if the pair
$(\{x\},\{y\})$ is contained in the greatest fixpoint of the
bisimulation function on the determinized transition system; (iv)
\emph{games:} there is an attacker-defender game characterizing the
equivalence.

Our focus is on~(ii) and~(iii), and our take is quite different from
the usual approach: Instead of first defining the behavioural
equivalence/metric and then setting up an expressive logic for it, we
start by defining the logic and derive a fixpoint equation from the
logic. Fixpoint equations are of great
interest %
since efficient algorithms for computing behavioural conformances are
almost always based on fixpoint characterizations; in future work,
we aim to exploit such characterizations for algorithmic purposes.
However, for a given logic, corresponding fixpoint equations do not
always exist, and we give conditions for ensuring that they do.
We use the Galois connection approach from
\cite{bgkm:hennessy-milner-galois} as a starting point, and instead of
instantiating it for each case study, we propose a generic coalgebraic
framework. By employing fibrations (resp.\ indexed categories)
\cite{jacobs-fibrations,Jacobs_indexedcat}, we parameterize over the
notion of \emph{conformance} (e.g.\ equivalence, metric) that is our
focus of attention. Moreover, we parametrize over a quantale in which
both conformances and formulae take their values.

One interest, particularly, is in linear-time notions of conformance (such as
trace/language equivalences and their quantitative cousins), for which
we work in an Eilenberg-Moore category where the coalgebras live. We
exploit the generalized powerset construction~\cite{jss:trace-determinization,sbbr:powerset-coalgebraically} and
characterize those (trace) logics that can be turned into suitable
fixpoint equations on the determinized coalgebra, using a notion of
compatibility \cite{bgkm:hennessy-milner-galois} that has its roots in
up-to techniques \cite{p:complete-lattices-up-to}. We also study the
relation of compatibility to the notion of depth-1 separation used in
(quantitative) graded logics
\cite{DorschEA19,fswbgkm:graded-logics-em}.

More concretely, we work with coalgebras of the form
$c\colon X\to FTX$, living in some category $\mathbf{C}$, where a
monad $T$ intuitively specifies the implicit branching (or side
effects) and a functor $F$ describes the explicit branching type. For
instance, for a non-deterministic automaton we choose $T=\power$ and
$F = \_^\Sigma\times \mathbf{2}$. We fix an \emph{EM-law}
$\zeta\colon TF\Rightarrow FT$ \cite{jss:trace-determinization}
allowing us to obtain a determinized coalgebra, i.e., a coalgebra of
the form $c^\#\colon TX\to FTX$ that can be viewed as a coalgebra in
the Eilenberg-Moore category of $T$.  We can then define a logic
function $\lo$ that is defined on sets of $\V$-valued predicates on
$X$ and whose least fixpoint induces a behavioural conformance on
$TX$. Alternatively, given the determinization $c^\#$, we can -- in a
fibrational style -- define a conformance on $TX$ as the greatest
fixpoint of a Kantorovich lifting followed by a reindexing via
$c^\#$. The aim is to show that both conformances coincide.

We allow arbitrary constants in the logic, which -- in particular in
the linear-time case -- are able to add extra distinguishing power to
the logic. Along the way we give an answer to the question of why,
unlike branching-time logics, linear-time logics often do not need
any additional (boolean) operators, only modalities and
constants.

As examples we consider bisimilarity and branching-time pseudometrics
for probabilistic transition systems, as well as linear-time
conformances such as trace equivalence and trace distance.

\smallskip

\noindent\emph{Roadmap.} After reviewing preliminaries
in \cref{sec:preliminaries}, we summarize the approach based on Galois
connections (adjunctions) in \cref{sec:adjoint-logic}. The
instantiation to generic coalgebras is presented in
\cref{sec:adjointlogic} and a concrete quantale-valued branching-time
logic spelled out in
\cref{sec:quantale-branching-logic}. \cref{sec:adjointlogicEM}
specializes to coalgebras in Eilenberg-Moore categories, leading to
results strengthening those for the general case. Finally,
\cref{sec:examples} details the linear-time case studies mentioned
above, and we conclude in \cref{sec:conclusion}.

\section{Preliminaries}
\label{sec:preliminaries}

We recall some basic definitions and facts on lattices, quantales,
generalized distances, coalgebra, monads and their algebras and on
indexed categories. We do assume basic familiarity with category
theory (e.g.~\cite{AHS90}).

\subsection{Lattices, Fixpoints and Galois Connections}

A \emph{complete lattice} $(\mathbb{L},\sqsubseteq)$ consists of a set
$\mathbb{L}$ with a partial order $\sqsubseteq$ such that each
$Y\subseteq \mathbb{L}$ has a least upper bound $\bigsqcup Y$ (also
called supremum, join) and a greatest lower bound $\bigsqcap Y$ (also
called infimum, meet).
The Knaster-Tarski theorem \cite{t:lattice-fixed-point} guarantees
that any monotone function $f\colon\mathbb{L}\to\mathbb{L}$ on a
complete lattice $\mathbb{L}$ has a \emph{least fixpoint} $\mu f$ and
a \emph{greatest fixpoint} $\nu f$.

Let $\mathbb{L}$, $\mathbb{B}$ be two lattices.  A \emph{Galois
  connection} from $\mathbb{L}$ to $\mathbb{B}$ is a pair
$\alpha\dashv \gamma$ of monotone functions
$\alpha\colon \mathbb{L}\to\mathbb{B}$,
$\gamma\colon \mathbb{B}\to\mathbb{L}$ such that for all $\ell\in L$,
$b\in\mathbb{B}$:
$\alpha(\ell)\sqsubseteq b \iff \ell\sqsubseteq \gamma(b)$.

A \emph{closure} $\cl\colon \mathbb{L}\to\mathbb{L}$ is a monotone,
idempotent and extensive (i.e.\
$\forall_{x\in\mathbb{L}}\ x \sqsubseteq \cl(x)$) function on a
lattice. A \emph{co-closure} is monotone, idempotent and extensive
wrt.\ $\sqsupseteq$. Given a Galois connection $\alpha\dashv\gamma$,
$\gamma\circ\alpha$ is always a closure and $\alpha\circ \gamma$ a
co-closure.

\subsection{Quantales and Generalized Distances}

\begin{definition}
  \label{def:quantale}
  A \emph{(unital, commutative) quantale} $(\V,\otimes,1)$, or
  just~$\V$, is a complete lattice with an associative, commutative
  operation $\otimes:\V\times\V\to\V$ with unit $1$ that distributes
  over arbitrary (possibly infinite) joins~$\bigvee$. If~$1$ is the
  top element of~$\V$, then~$\V$ is \emph{integral}.
\end{definition}

\noindent In a quantale~$\V$, the functor $-\otimes y$ has a right
adjoint $[y,-]$ for every $y\in\V$; that is,
$x\otimes y\le z \iff x\le [y,z]$ for all $x,y,z\in\V$.

\begin{example}
  \label{ex:quantale}
  \begin{enumerate}
  \item \label{ex:quantale-boolean}The Boolean algebra $\mathbf{2}$
    with $\otimes=\wedge$ and unit~$1$ is an integral quantale; for
    $y,z\in\mathbf{2}$, we have $[y,z]=y\to z$.
  \item \label{ex:quantale-reals}The complete lattice $[0,1]$ ordered
    by the reversed order of the reals, i.e.~${\le}={\ge_\mathbb{R}}$,
    and equipped with truncated addition $r\otimes s=\min(r+s,1)$, is
    an integral quantale; for $r,s\in[0,1]$, we have
    $[r,s]=s\dotminus r = \max(s-r, 0)$ (truncated subtraction).
\end{enumerate}
\end{example}

\begin{toappendix}
  \begin{lemmarep}
    \label{lem:internal-hom-unit}
    Let $\mathcal{V}$ be a quantale. Then for $z\in V$ it holds
    that $[1,z] = z$. If in addition the quantale is integral, we have
    $[z,1] = 1$
  \end{lemmarep}

  \begin{proof}
    As right adjoint of the associative operation the internal hom of
    a quantale is concretely defined as follows:
    \[ [y,z] = \bigvee \{ x'\in \V\mid x'\otimes y \le z \}. \]
    Hence we can simply use this characterization and obtain:
    \[ [1,z] = \bigvee \{ x'\in \V\mid x'\otimes 1 \le z \} = \bigvee
      \{ x'\in \V\mid x' \le z \} = z. \]
    If in addition $1$ is the largest element, we have:
    \[ [z,1] = \bigvee \{ x'\in \V\mid x'\otimes z \le 1 \} = \bigvee
      \V = 1. \qedhere \]
  \end{proof}
\end{toappendix}

\noindent Given a quantale $\V$, a \emph{directed ($\V$-valued) pseudometric}
(\emph{on $X$}) is a function $d\colon X\times X\to \V$ such that
(i)~$\forall_{x\in X}\ d(x,x) \ge 1$ (reflexivity);
(ii)~$\forall_{x,y,z\in X}\ d(x,z)\ge d(x,y)\otimes d(y,z)$
(transitivity/triangle inequality). Moreover,~$d$ is a
\emph{pseudometric} if additionally,
(iii)~$\forall_{x,y\in X}\ d(x,y) = d(y,x)$ (symmetry). We write
$\dpmet_\V(X)$ to denote the lattice of all directed pseudometrics on~$X$, while for pseudometrics we use $\pmet_\V(X)$.  Given
$d_X\in\dpmet_\V(X)$, $d_Y\in\dpmet_\V(Y)$, a function
$f\colon X\to Y$ is \emph{non-expansive} (wrt.\ $d_X,d_Y$) if
$d_X(x,x') \le d_Y(f(x),f(x'))$ for all $x,x'\in X$. Note that due to
the choice of order in the quantale, the inequality in the definitions
above is reversed wrt.\ the standard definitions that are typically given in
the order on the reals. As originally observed by
Lawvere~\cite{Law73}, one may see directed $\V$-valued pseudometrics
as $\V$-enriched categories, or just $\V$-categories, and
non-expansive functions as $\V$-functors.

\begin{toappendix}
  \begin{lemmarep}
    \label{lem:non-expansive-function-from-d}
    Let $d\in\dpmet_\V(X)$ and define
    $d_\V = [\_,\_]\colon \V\times\V\to\V$. Fix $y\in X$ and define a
    function $h\colon X\to\V$ as $h(x) = d(y,x)$. Then $h$ is
    non-expansive wrt.\ $d,d_\V$.

    The same holds if $d\in\pmet_\V(X)$ and $d_\V$ is defined as
    $d_\V(x,y) = [x,y]\land [y,x]$.

    Whenever $\V$ is integral, in both cases it holds for all $x\in X$
    that $d_\V(h(y),h(x)) = d(y,x)$
  \end{lemmarep}

  \begin{proof}
    We start with a directed metric $d\in\dpmet_\V(X)$ and
    $d_\V = [\_,\_]$ and our aim is to show that
    $d\preceq d_\V\circ (h\times h)$:
    \begin{align*}
      d\preceq d_\V\circ (h\times h) &\iff
      \forall_{x_1,x_2\in X}\  d(x_1,x_2) \le d_\V(h(x_1),h(x_2)) \\
      &\iff
      \forall_{x_1,x_2\in X}\  d(x_1,x_2) \le [h(x_1),h(x_2)] \\
      &\iff
      \forall_{x_1,x_2\in X}\  d(x_1,x_2)\otimes h(x_1) \le h(x_2) \\
      &\iff
      \forall_{x_1,x_2\in X}\  d(x_1,x_2)\otimes d(y,x_1) \le d(y,x_2) \\
      &\iff
      \forall_{x_1,x_2\in X}\  d(y,x_1)\otimes d(x_1,x_2) \le d(y,x_2)
    \end{align*}
    Here we use the fact that $x\otimes y\le z \iff x\le [y,z]$
    (adjunction) and the commutativity of $\otimes$. The last
    statement is true due to the triangle inequality and hence
    non-expansiveness of $h$ follows.

    Whenever $d_\tobject = [\_,\_]$, $d_\V(h(y),h(x))$ simplifies to
    $d_\V(d(y,y),d(y,x)) = [1,d(y,x)] = d(y,x)$, using the fact that
    $\V$ is integral (and hence $d(y,y) = 1$) and
    \cref{lem:internal-hom-unit}.

    We now take an undirected metric $d\in\pmet_\V(X)$ and define
    $d_\V(x,y) = [x,y]\land [y,x]$. Again our aim is to show that
    $d\preceq d_\V\circ (h\times h)$:
    \begin{align*}
      &d\preceq d_\V\circ (h\times h) \\
      &\iff
      \forall_{x_1,x_2\in X}\  d(x_1,x_2) \le d_\V(h(x_1),h(x_2)) \\
      &\iff
      \forall_{x_1,x_2\in X}\  d(x_1,x_2) \le [h(x_1),h(x_2)]\land[h(x_2),h(x_1)] \\
      &\iff
      \forall_{x_1,x_2\in X}\  d(x_1,x_2) \le [h(x_1),h(x_2)] \text{ and
        }d(x_1,x_2) \le [h(x_2),h(x_1)] \\
      &\iff
      \forall_{x_1,x_2\in X}\  d(x_1,x_2)\otimes h(x_1) \le h(x_2)
      \text{ and }d(x_1,x_2)\otimes h(x_2) \le h(x_1)\\
      &\iff
      \forall_{x_1,x_2\in X}\  d(x_1,x_2)\otimes d(y,x_1) \le d(y,x_2)
      \text{ and }d(x_1,x_2)\otimes d(y,x_2) \le d(y,x_1) \\
      &\iff
      \forall_{x_1,x_2\in X}\  d(y,x_1)\otimes d(x_1,x_2) \le d(y,x_2)
      \text{ and } d(y,x_2)\otimes d(x_2,x_1) \le d(y,x_1)
    \end{align*}
    Here we use again the adjunction and the commutativity of
    $\otimes$. In addition we need symmetry of $d$. The last
    statements are again true due to the triangle inequality and hence
    non-expansiveness of $h$ follows.

    In this case $d_\V(h(y),h(x))$ simplifies to the undirected case
    we have
    $d_\V(d(y,y),d(y,x)) = [1,d(y,x)] \land [d(x,y),1] = d(y,x) \land
    1 = d(y,x)$, based on \cref{lem:internal-hom-unit}, where we again
    use that the quantale is integral.
  \end{proof}
\end{toappendix}

\subsection{Coalgebras and Eilenberg-Moore Categories}
\label{subsec:coalgebras-EM}

Given a functor $F\colon\Ccat\to\Ccat$, an \emph{$F$-coalgebra}
$(X,c)$ (or simply $c$) consists of an object $X\in\Ccat$ and a
$\Ccat$-arrow $c\colon X\to FX$. In the paradigm of \emph{universal
  coalgebra}~\cite{Rutten00}, we understand~$X$ as the state space of
a transition system, $F$ as specifying the branching type of the
system, and~$c$ as a transition map that assigns to each state a
collection of successors structured according to~$F$. For instance,
when $\Ccat=\set$ is the category of sets and functions, then the
powerset functor $\power$ assigns to each set its powerset, and
$\power$-coalgebras are just sets equipped with a transition
relation. On the other hand, the (finitely supported) distribution
functor~$\dist$ assigns to each set~$X$ the set of finitely supported
probability distributions on~$X$, given in terms of maps
$p\colon X\to [0,1]$ with finite support such that
$\sum_{x\in X}p(x)=1$. A $\dist$-coalgebra thus is precisely a
Markov chain.

Recall that a \emph{monad}
$(T,\eta\colon \text{Id}\Rightarrow T, \mu\colon TT\Rightarrow T)$ on
$\Ccat$, usually denoted by just~$T$, consists of a functor
$T\colon\Ccat\to\Ccat$ and natural transformations
$\eta\colon \text{Id}\Rightarrow T$ (the \emph{unit}) and
$\mu\colon TT\Rightarrow T$ (the \emph{multiplication}), subject to
certain coherence laws. Monads abstractly capture notions of algebraic
theory, with~$TX$ being thought of as terms modulo provable equality
over variables in~$X$. Correspondingly, a \emph{$T$-algebra} $(X,a)$
consists of an object $X$ of $\Ccat$ and a $\Ccat$-arrow
$a\colon TX\to X$ such that $a\circ \eta_X = \mathit{id}_X$ and
$a\circ Ta = a\circ \mu_X$; we may think of $T$-algebras as algebras
for the algebraic theory encapsulated by~$T$. A homomorphism between
$T$-algebras $(X,a)$, $(Y,b)$ is a $\Ccat$-arrow $f\colon X\to Y$ such
that $b\circ Tf = f\circ a$. The \emph{Eilenberg-Moore category}
of~$T$, denoted $\alg{T}$, is the category of $T$-algebras and their
homomorphisms.  There is a free-forgetful adjunction
$L\dashv R\colon \Ccat \to \alg{T}$, where the forgetful functor $R$
maps an algebra to its underlying $\Ccat$-object and $L$ maps an
object $X\in\Ccat$ to the free algebra $(TX,\mu_X)$.

A \emph{(monad-over-functor) distributive law (or EM-law)} of a
monad~$T$ over a functor~$F$ is a natural transformation
$\zeta\colon TF\Rightarrow FT$ satisfying
$\zeta_X\circ \eta_{FX} = F\eta_X$ and
$\zeta_X\circ \mu_{FX} = F\mu_X \circ \zeta_{TX} \circ T\zeta_X$. This
is equivalent to saying that the assignment
$\tilde{F}(X,a) = (FX,Fa\circ \zeta_X)$ defines a \emph{lifting}
$\tilde F\colon \alg{T} \to \alg{T}$ of~$F$ (where \emph{lifting}
means that
$R\tilde F=FR$). Then, the \emph{determinization}
\cite{sbbr:powerset-coalgebraically} of a coalgebra $c\colon X\to FTX$
in $\Ccat$ is the transpose $c^\# \colon LX \to \tilde F LX$ of~$c$
under $L\dashv R$. More concretely the determinization can be obtained
as $c^\# = F\mu_X \circ \zeta_{TX}\circ Tc$. For instance, when
$FX= X^\Sigma\times \mathbf{2}$ and $T=\power$, then this yields
exactly the standard powerset construction for the determinization
of non-deterministic automata.

\subsection{Indexed Categories and Fibrations}

Our aim is to equip objects of a category with additional information,
e.g., consider sets with (equivalence) relations or metrics. Formally,
this is done by working with fibrations, in particular we will
consider fibrations arising from the Grothendieck construction for
indexed categories \cite{jacobs-fibrations,Jacobs_indexedcat}. For us
it is sufficient to consider as indexed categories functors
$\Phi\colon \Ccat^{\text{op}}\to \pos$, where $\pos$ is the category
of posets (ordered by $\preceq$) with monotone maps.  Such
functors induce a fibration $U\colon \elements{\Phi} \to \Ccat$ where
$U$ is the forgetful functor and $\elements{\Phi}$ is the category
whose objects and arrows are characterized as follows:
\[
  \infer={(X,d) \in \elements \Phi}{X \in \Ccat \ \land \ d\in\Phi X}\qquad
  \infer={(X,d) \xrightarrow f (Y,d') \in \elements\Phi}{X \xrightarrow f Y \in \Ccat \ \land \ d \preceq (\Phi f) d'}
\]
\noindent Here, $f^*=\Phi f$ is also called the \emph{reindexing operation},
and $d$ is called a \emph{conformance}.

Typical examples are functors $\Phi\colon\set^{\text{op}}\to \pos$
mapping a set $X$ to the lattice of equivalence relations or
pseudometrics on $X$.

\section{Adjoint Logic: the General Framework}
\label{sec:adjoint-logic}

\noindent We summarize previous results on relating logical and
behavioural conformances using Galois
connections~\cite{bgkm:hennessy-milner-galois}. These results are
based on the following well-known property that shows how fixpoints
are preserved by Galois connections
(e.g.~\cite{bkp:abstraction-up-to-games-fixpoint,cc:systematic-analysis,cc:temporal-abstract-interpretation}). The
formulation of these properties involves a notion of compatibility
studied for coinductive up-to techniques
\cite{p:complete-lattices-up-to}.

\begin{definition}
  Let $\lo,\cl\colon \mathbb{L}\to\mathbb{L}$ be monotone
  endofunctions on a partial order $(\mathbb{L},\sqsubseteq)$. Then
  $\lo$ is \emph{$\cl$-compatible} if
  $\lo\circ \cl \sqsubseteq \cl\circ \lo$.
\end{definition}

\begin{theorem}[\cite{bgkm:hennessy-milner-galois}]
  \label{thm:fixpoint-preservation}
  Let $\alpha\colon \mathbb{L} \to \mathbb{B}$,
  $\gamma\colon \mathbb{B} \to \mathbb{L}$ be a Galois connection
  between complete lattices $\mathbb{L}, \mathbb{B}$, and
  let\/ $\lo\colon \mathbb{L}\to\mathbb{L}$,
  $\be\colon \mathbb{B}\to\mathbb{B}$ be monotone. Then the following
  holds.
  \begin{enumerate}
  \item If $\be = \alpha\circ \lo\circ \gamma$ then
    $\alpha(\mu \lo) \sqsubseteq \mu \be$.
  \item If $\alpha\circ \lo = \be \circ \alpha$, then
    $\alpha(\mu\,\lo) = \mu\,\be$. If\/ $\lo$ reaches its fixpoint
    in~$\omega$ steps, i.e., $\mu \lo = \lo^\omega(\bot)$, then so
    does\/ $\be$.
  \item Let\/ $\cl = \gamma\circ \alpha$ be the closure operator of the
    Galois connection, and suppose that\/
    $\be = \alpha\circ \lo\circ \gamma$. If\/ $\lo$ is
    $\cl$-compatible, then $\alpha(\mu\,\lo) = \mu\,\be$.
  \end{enumerate}
\end{theorem}

\begin{remark}
  \label{rem:weaker-compatibility}
  In fact, there is a weaker notion than compatibility that ensures the
  same result, i.e., $\alpha(\mu\,\lo) = \mu\,\be$. In particular, it
  is sufficient to show the following condition:
  \begin{itemize}
  \item $\lo(\cl(\ell)) \sqsubseteq \cl(\lo(\ell))$ for all
    $\ell\in \mathbb{S}\subseteq \mathbb{L}$ where $\mathbb{S}$ is an
    invariant of $\lo$, i.e.\
    $\lo[\mathbb{S}]\subseteq \mathbb{S}$,
    $\bot\in\mathbb{S}$, and $\mathbb{S}$ is closed under directed
    joins. (If the fixpoint is reached in $\omega$ steps, closure under
    directed joins is not required.)
  \end{itemize}
\end{remark}

\noindent We apply this in a scenario where~$\mathbb{L}$ consists of
logical formulas (or more precisely their semantics, in the shape of
sets of definable predicates) and~$\mathbb{B}$ consists of
conformances. These Galois connections will be contra-variant when we
consider the quantalic ordering in $\mathbb{B}$. Then, $\lo$ is the
``logic function'' that adds a layer of modalities and propositional
operators to a given set of predicates, so that $\mu\,\lo$ is the
semantics of the set of formulas of the logic. On the other hand,
$\be$ is the ``behaviour function'', whose greatest fixpoint $\nu\be$
(remember the contra-variance) is behavioural conformance.

\begin{example}
  \label{ex:Galois-connections}
  The simplest Galois connection used in
  \cite{bgkm:hennessy-milner-galois} for characterizing behavioural
  equivalence is between $\mathbb{L} = \power(\mathbf{2}^X)$ (sets of
  predicates on $X$) and $\mathbb{B} = \pmet_{\mathbf{2}}(X)$
  (equivalences on $X$), where $\alpha$ maps every set of predicates
  to the equivalence relation induced by it and $\gamma$ maps an
  equivalence to the set of predicates closed under it.

  Moving to pseudometrics we obtain a Galois connection between
  $\mathbb{L} = \power([0,1]^X)$ (sets of real-valued predicates on
  $X$) and $\mathbb{B} = \pmet_{[0,1]}(X)$ (pseudometrics on $X$)
  where $\alpha$ maps every set of functions $X\to[0,1]$ to the least
  pseudometric making all these functions non-expansive and $\gamma$
  takes a pseudometric and produces all its non-expansive functions.
\end{example}

So first, define a logical universe $\mathbb{L}$ and a logic function
$\lo\colon\mathbb{L}\to \mathbb{L}$. Second, choose a suitable Galois
connection $\alpha\dashv \gamma$ to a behaviour universe $\mathbb{B}$
and show that $\lo$ is $\cl$-compatible.
Third, derive the behaviour function
$\be = \alpha\circ\lo\circ\gamma\colon \mathbb{B}\to\mathbb{B}$.  From
the results above, we automatically obtain the equality
$\alpha(\mu\,\lo) = \mu \,\be$, which tells us that logical and
behavioural equivalence respectively distance coincide
(Hennessy-Milner theorem).

\begin{toappendix}
  Combining logic functions results in the combination of the
  corresponding behaviour functions, which is essential in
  establishing Hennessy-Milner theorems compositionally.

  \begin{proposition}[\cite{bgkm:hennessy-milner-galois}]
    \label{prop:compositionality}
    Let $\lo_1,\lo_2, \cl\colon \mathbb{L}\to\mathbb{L}$ be monotone
    functions on a complete lattice $\mathbb{L}$ such that $\lo_1$,
    $\lo_2$ are $\cl$-compatible. Then\/ $\lo_1\sqcup \lo_2$ and
    $\lo_1\circ\lo_2$ are also $\cl$-compatible.

    Further, let\/ $\be_i = \alpha\circ\lo_i\circ\gamma$ be the
    induced behaviour functions ($i=1,2$). Then the behaviour
    functions of\/ $\lo_1\sqcup \lo_2$ and\/ $\lo_1\circ\lo_2$ are\/
    $\be_1\sqcup\be_2$ and\/ $\be_1\circ\be_2$, respectively.  Every
    constant function $k$ and the identity are $c$-compatible. Their
    behaviour functions are the constant function
    $b\mapsto \alpha(\ell)$ (where $\ell$ is the constant value of
    $k$) respectively the co-closure $\alpha\circ\gamma$.
  \end{proposition}
\end{toappendix}

\section{Adjoint Logic for Coalgebras}
\label{sec:adjointlogic}

In this section we will describe a general framework where the adjoint
logic is instantiated to the setting of coalgebraic modal logic. %

\subsection{Setting up the Adjunction}

One can generalize from \cref{ex:Galois-connections} and instead of a
set $X$ take an object in a (locally small) category $\Ccat$.
Furthermore we fix an object $\tobject\in\Ccat$ (the \emph{truth
  value} object), which in all our applications will be a
quantale. Predicates are represented by the indexed
category
$\Ccat(\_,\tobject)$; thus, sets of predicates (lattice $\mathbb{L}$)
are given by the indexed category $\power \circ \Ccat(\_,\tobject)$
(where the order is inclusion). In addition, we use an indexed
category~$\Phi$ specifying the notion of conformance on~$X$ (lattice
$\mathbb{B}$) and work with the following assumptions:
\begin{enumerate}[label=\textbf{A\arabic*}]
 \setcounter{enumi}{0}
 \item\label{assum:ConfExists}  %
 Each fibre $\Phi X$ is a poset having
arbitrary meets (thus, a complete lattice) and the reindexing map
preserves these meets (i.e.\ $\elements{\Phi}$ has fibred limits).
 \item\label{assum:TruthConfExists}  %
 Let $d_{\tobject}\in\Phi\tobject$ be a fixed conformance on the
truth value object $\tobject$.
\end{enumerate}

For an arrow $f \in \Ccat(X,Y)$, we write $f^\bullet$ for the
reindexing in $\Ccat(\_,\tobject)$ ($f^\bullet g = g\circ f$, where
$g\in \Ccat(Y,\tobject)$) and $f^*$ for the reindexing in $\Phi$
($f^* = \Phi f$).

\begin{theoremrep}
  \label{thm:fibredadjoint}
  Let $X$ be an object of $\mathbf{C}$. Under
  Assumptions~\ref{assum:ConfExists} and~\ref{assum:TruthConfExists},
  there is a dual adjoint situation (contravariant Galois connection)
  $\alpha_X \dashv \gamma_X$ between the underlying fibres:
  \begin{align*}
    \alpha_X\colon \power (\mathbf{C} (X,\tobject)) \to
    \Phi(X)^{\text{op}} &\quad S \subseteq
    \mathbf{C}(X,\tobject) \mapsto \bigwedge_{k\in S} k^*(d_\tobject)\\
    \gamma_X\colon \Phi(X)^{\text{op}} \to \power (\mathbf{C}
    (X,\tobject)) & \quad d\in\Phi X \mapsto
    \{k\in\mathbf{C}(X,\tobject)\mid d\preceq k^*(d_\tobject) \}.
  \end{align*}
  More concretely: $\alpha_X,\gamma_X$ are both antitone
  ($S\subseteq S'\implies \alpha_X(S)\succeq \alpha_X(S')$,
  $d\preceq d'\implies \gamma_X(d) \supseteq \gamma_X(d')$) and we
  have $d\preceq \alpha_X(S) \iff S\subseteq \gamma_X(d)$ for
  $d\in \Phi X$ and $S \in \power\Ccat(X,\tobject)$.
\end{theoremrep}

\begin{proof}
  First we argue that $\alpha_X,\gamma_X$ are order preserving (in
  this case antitone). For $\alpha_X$, let $S\subseteq S'$. Clearly,
  \[
  \alpha_X (S) = \bigwedge_{h\in S} h^* d_\tobject \succeq \bigwedge_{h\in S'} h^* d_{\tobject} = \alpha_X(S').
  \]
  For $\gamma_X$, let $d \preceq d'$ and let $h\in \gamma_X (d')$. I.e.,
  $d' \preceq h^* d_\tobject$; thus, $d \preceq h^* d_\tobject$. Hence, $h\in \gamma_X (d)$.

  Now it remains to show that $d \preceq \alpha_X (S) \iff S \subseteq
  \gamma_X (d)$, for every $d\in \Phi X$ and $S \in \power\Ccat(X,\tobject)$.
  \begin{align*}
    d\preceq \alpha_X (S) \iff&\ d \preceq \bigwedge_{h\in S} h^*d_\tobject
    \iff \ \forall_{h\in S}\ d \preceq h^*d_\tobject\\
    \iff&\ \forall_{h\in S}\ h\in \gamma_X(d)
    \iff \ S \subseteq \gamma_X(d). \qedhere
  \end{align*}
\end{proof}
\noindent Thus, for $X\in\Ccat$, the fibres $\power \Ccat(X,\tobject)$
and $(\Phi X)^{\text{op}}$ will take the role of $\mathbb B$ and
$\mathbb L$ (respectively) as in
\cite[Theorem~3.2]{bgkm:hennessy-milner-galois}. Moreover,
\cref{thm:fibredadjoint} will be instantiated to obtain the desired
Galois connections between predicates and conformances for our case
studies.

\begin{example}
  \label{ex:conformances}
  Let $\Ccat = \set$ and $\tobject=\V$ be a quantale. We consider
  $\Phi X = \dpmet_\V(X)$ (resp., $\Phi X = \pmet_\V(X)$) with the
  order $\preceq$ on $\Phi X$ induced by the pointwise lifting of the
  order $\leq$ on~$\V$. The reindexing functor $f^*$ for a function
  $f\colon X \to Y$ is given by $f^*d =d\circ (f\times f)$; thus,
  satisfying \ref{assum:ConfExists}. As conformance $d_\tobject$ on
  $\V$ we take the internal hom $[\_,\_]$ (resp., its symmetrization:
  $d_\tobject(x,y) = [x,y]\land [y,x]$); thus, satisfying
  \ref{assum:TruthConfExists}. Then we have
  $\alpha_X\dashv\gamma_X \colon \power(\set(X,\V)) \rightleftarrows
  \dpmet_\V(X)$, where:
  \begin{align*}
    \alpha_X(S)(x,x') =\ &  \bigwedge_{h\in S} d_\tobject(h(x),h(x')) \\
    \gamma_X(d)=\ & \left\{h\colon X\to \V \mid \forall_{x,x'\in X}\
      d(x,x') \leq d_\tobject(h(x),h(x')) \right\}
  \end{align*}

  In both cases, $\alpha_X$ assigns to a set of maps the
  greatest (directed) pseudometric making
  all these maps non-expansive, while $\gamma_X$ maps a pseudometric
  all its non-expansive maps.
\end{example}

\subsection{Characterizing Closure}

Given that the key condition imposed on the logic function in
\cref{thm:fixpoint-preservation} is compatibility with the closure of
the Galois connection, it is important to understand how this closure
operates. In the setting of \cref{thm:fibredadjoint} we can
characterize the closure in terms of non-expansive propositional
operators, provided that~$\gamma$ is natural. We note first
that~$\alpha$ is always natural:

\begin{proposition}
  In \cref{thm:fibredadjoint}, the transformation %
  $\alpha$ is natural in $X\in\mathbf C$, that is, for
  $f\in \mathbf{C}(X,Y)$, we have
  $\alpha_X \circ \power (f^\bullet) = f^*\circ \alpha_Y$.
\end{proposition}

\noindent For the right adjoint $\gamma$, naturality need not hold in
general. It does hold for $\set$ and generalized (directed) metrics
over the quantales in \cref{ex:quantale}. A counterexample can
however be constructed for $\Ccat = \alg{\power}$ and $\V=[0,1]$
(see \cite{bgkm:hennessy-milner-galois}).

If $\gamma$ is natural,
then we can characterize the closure $\gamma\circ \alpha$ using the
internal language of indexed categories. To this end, suppose that
$\mathbf C$ has (small) products. Then for every
$S\subseteq \mathbf C(X,\tobject)$. we have a unique tupling
$\langle S \rangle \colon X \to \tobject^S$ such that
$\pi_k \circ \langle S \rangle=k$ for all $k\in S$, where
$\pi_k\colon \tobject^S \to \tobject$ is the product projection
for~$k$.

\begin{lemmarep}
  \label{lem:charClosure}
  The right adjoint $\gamma$ in \cref{thm:fibredadjoint} is laxly
  natural, i.e.\
  $\power (f^\bullet) \circ \gamma_Y \subseteq \gamma_X \circ f^*$ for
  $f\colon X \to Y \in \mathbf C$.
  If $\gamma$ is natural (i.e.\ the inclusion is an equality) and
  $\mathbf C$ has products, then
  $\gamma_X(\alpha_X(S)) = \power(\langle S \rangle^\bullet)
  (\gamma_{\tobject^S} (d_{\tobject^S}))$, where
  $d_{\tobject^S} = \bigwedge_{k\in S} \pi_k^* d_\tobject$.
\end{lemmarep}

\begin{proof}
  In order to show $\power (f^\bullet) \circ \gamma_Y \subseteq
  \gamma_X \circ f^*$ we fix $d\in\Psi X$ and $h\in \power
  (f^\bullet)(\gamma_Y(d))$. This implies that $h = k\circ f$ where
  $k\in \gamma_Y(d))$, hence $d\preceq k^*d_\tobject$. From this we can
  infer $h^* d_\tobject = (k\circ f)^*d_\tobject = f^*(k^*d_{\tobject})
  \succeq f^*d$. Hence $h\in\gamma_X(f^*d)$ follows.

  Now suppose that $\gamma$ is natural. We observe that
  $\langle S \rangle^* d_\tobject^S = \alpha_X S$ for every
  $S\subseteq \mathbf C(X,\tobject)$.
  \begin{align*}
    \langle S \rangle^* d_{\tobject^S} =&\ \langle S \rangle^* \bigwedge_{k\in S} \pi_k^* d_\tobject && \because \text{$\Phi$ has indexed limits}\\
    =&\  \bigwedge_{k\in S} \langle S \rangle^* \pi_k^* d_\tobject && \because \text{reindexing $(g\circ f)^* = f^* \circ g^*$}\\
    =&\ \bigwedge_{k\in S} (\pi_k \circ \langle S \rangle)^*  d_\tobject && \because \text{universal property of products}\\
    =&\ \bigwedge_{k\in S} k^* d_\tobject = \alpha_X S.
  \end{align*}
  Now the result follows directly from the naturality of $\gamma$, i.e.,
  \[
    \power(\langle S \rangle^\bullet) \gamma_{\tobject^S}
    (d_{\tobject^S}) = \gamma_X \langle S \rangle^*
    (d_{\tobject^S}) = \gamma_X \alpha_X S. \qedhere
  \]
\end{proof}

\noindent This result can be interpreted as follows:
$\gamma_{\tobject^S} (d_{\tobject^S})$ is the set of all non-expansive
functions $\tobject^S\to \tobject$, hence all non-expansive operators
of arbitrary arity on $\tobject$. Reindexing via $\langle S \rangle$
means to combine all predicates in $S$ via those operators, hence we
describe the closure under all non-expansive operations on $\tobject$.

\subsection{Towards a Generic Logic Function}
\label{subsec:gen-logic}

Since our slogan is to generate the behaviour function from the logic
function, we start by setting up our logical framework first. %
Following
\cite{bgkm:hennessy-milner-galois,Forster_et_al:CSL.2023:Density}, we
adopt a semantic approach to defining a (modal) logic, i.e., we
specify the operators (including modalities) as a transformation of
predicates; formally, as a (natural) transformation $\lo$ of type
$\power\Ccat(\_,\tobject) \Rightarrow
\power\Ccat(\_,\tobject)$. The idea is that the logic function
$\lo_X$ adds one ``layer'' of modal depth; in particular, the least
fixpoint $\mu \lo_X$ of $\lo_X$ can be seen as the set of
(interpretations of) all modal formulas.

In particular, we require
\begin{enumerate}[label=\textbf{A\arabic*}]
 \setcounter{enumi}{2}
 \item\label{assum:EvalsExist} Fix a family
$(\ev_{\lambda}\in\Ccat)_{\lambda\in\Lambda}$ of evaluation maps
$\ev_{\lambda}\colon F\tobject\to \tobject$.
\end{enumerate}
As noted in \cite{Schroeder2007:ExpressivityCoalgModalLog}, such evaluation maps -- commonly used in
coalgebraic modal logic -- correspond to
natural transformations of type $\Ccat(\_,\tobject)\to\Ccat(F\_,\tobject)$ by the Yoneda lemma.

\begin{propositionrep}
  A family of evaluation maps
  $(\ev_{\lambda}%
  )_{\lambda\in\Lambda}$ induces a natural transformation %
  $\Lambda \colon \power\Ccat(\_,\tobject) \Rightarrow
  \power\Ccat(F\_,\tobject)$ given by
  $S \mapsto \{\lambda_X(h)\mid \lambda\in\Lambda, h\in S\}$, where
  $\lambda_X(h)=\ev_\lambda \circ Fh$.
\end{propositionrep}

\begin{proof}
  Naturality of $\Lambda$ amounts to showing that, given $f\colon X\to
  Y$ in $\Ccat$ and $S'\subseteq \Ccat(X,\tobject)$:
  \begin{align*}
    &\power (Ff)^\bullet (\Lambda_Y S') = \power (Ff)^\bullet
    \{\ev_\lambda \circ Fh \mid h\in S'\} = \{\ev_\lambda \circ Fh \circ
    Ff \mid h\in S'\} \\
    = &\{\ev_\lambda \circ F(h \circ f) \mid h\in S'\} =
    \{\ev_\lambda \circ Fk \mid k\in \power f^\bullet(S')\} =
    \Lambda_X(\power f^\bullet(S')) \qedhere
  \end{align*}
\end{proof}

Apart from modalities, a logic typically needs operators and
constants. We do not consider constants as $0$-ary operators, which
allows us to distinguish between operators that arise as in
\cref{lem:charClosure} from the closure of the Galois connection (that
is, non-expansive operators) and the remaining (constant) operators
that bring additional distinguishing power. This is, for instance,
needed in the case of trace equivalence on a determinized transition
system to distinguish the empty set of states from sets of
states having no transitions. We need an additional (constant)
predicate for this task that can neither be provided by the closure nor
by a constant modality (see \cref{subsec:trace-equivalence}).

\begin{enumerate}[label=\textbf{A\arabic*}]
 \setcounter{enumi}{3}
\item\label{assum:ConstantsExist} We assume a set
  $\Theta_X \subseteq \Ccat(X,\tobject)$ of \emph{constants} (which is
  later restricted to consist of free extensions of constant
  maps).
\end{enumerate}

To model the propositional operators, we introduce a closure $\cl'_X$:
\begin{enumerate}[label=\textbf{A\arabic*}]
 \setcounter{enumi}{4}
\item\label{assum:PropositionalClosure} For each $X\in\Ccat$ we assume
  that there is a closure
  $\cl'_X \colon \power \Ccat (X,\tobject) \to \power
  \Ccat(X,\tobject)$ (not necessarily natural), specifying the
  propositional operators.
\end{enumerate}

We say
that $\cl'_X$ is a \emph{subclosure} of $\cl_X$ whenever
$\cl'_X\subseteq \cl_X$,
which means that the propositional operators implemented by $\cl'_X$ are
already contained in the closure induced by the Galois connection
(cf.~\cref{lem:charClosure}).

\begin{toappendix}
  \begin{lemmarep}
    \label{lem:char-subclosure}
    Let $\cl_X = \gamma_X\circ\alpha_X$ be the closure obtained from
    the Galois connection in \cref{thm:fibredadjoint} and let
    $\cl'_X\colon \power \Ccat (X,\tobject) \to \power
    \Ccat(X,\tobject)$ be some closure. Then $\cl'_X\subseteq \cl_X$
    is equivalent to $\cl_X = \cl'_X\circ\cl_X$ and equivalent to
    $\cl_X = \cl_X\circ\cl'_X$.
  \end{lemmarep}

  \begin{proof}
    Assume that $\cl'_X\subseteq \cl_X$. We obtain
    $\cl_X\subseteq \cl'_X\circ\cl_X$,
    $\cl_X\subseteq \cl_X\circ\cl'_X$ trivially, since $\cl'_X$ is a
    closure and satisfies
    $\mathit{id}_{\Ccat (X,\tobject)}\subseteq \cl'_X$. Furthermore
    \[ \cl'_X\circ\cl_X \subseteq \cl_X\circ\cl_X = \cl_X, \]
    similarly for the other equality ($\cl_X = \cl_X\circ\cl'_X$).

    Now assume that $\cl'_X\circ\cl_X = \cl_X$. Then, again by using
    the fact that $\cl_X$ is a closure:
    \[ \cl'_X \subseteq \cl'_X\circ\cl_X = \cl_X. \]

    If on the other hand $\cl_X = \cl_X\circ\cl'_X$ we can argue
    similarly by observing
    \[ \cl'_X \subseteq \cl_X\circ\cl'_X = \cl_X. \qedhere \]
  \end{proof}
\end{toappendix}

Now we can define the logic function for a coalgebra $c$ as
\[
  \lo_X = \power(c^\bullet)\circ \Lambda_X\circ \cl'_X \cup\
  \Theta_X\colon \power \Ccat(X,\tobject)\to \power\Ccat(X,\tobject).
\]
Its least fixpoint contains all predicates that can be described by
modal formulas.\footnote{In our setup a
  formula is either a constant or starts with a modality, which still
  results in an expressive logic. One could slightly modify $\lo_X$
  and obtain all formulas by adding another closure $\cl'$.}

\begin{example}
  \label{ex:evaluation-maps-logic}
  Let $\Ccat=\set$, $c\colon X\to FX$ and $\Phi X= \pmet_\V(X)$.
  Recall the Galois connection from \cref{ex:conformances} and
  consider the following two examples, where in both cases no
  constants are needed, i.e., $\Theta_X$ is empty; thus,
  \ref{assum:ConstantsExist} vacuously holds.
  \begin{enumerate}[wide]
  \item \emph{Bisimilarity on (unlabelled) transition systems:} we let
    $F=\power_\mathit{fin}$ (finite powerset functor), $\V=2$, and
    consider the evaluation map
    $\ev_\Diamond\colon\power_\mathit{fin} 2 \to 2$ encoding the usual
    diamond modality: $\ev_\Diamond(U)=1 \iff 1\in U$. This can be
    extended to a logic by choosing as $\cl'$
    (Assumption~\ref{assum:PropositionalClosure}) the closure under
    all (finitary) Boolean operators.
  \item\label{item:prob} \emph{Behavioural metrics for probabilistic
      transition systems with termination:} we let $FX=\dist X+1$
    (where $1=\{\checkmark\}$) and
    $\V=(\left[0,1\right],\geq_\mathbb{R})$. Define two evaluation
    maps: $\ev_E\colon D\left[0,1\right]+1 \to \left[0,1\right]$
    corresponds to expectation, i.e.\
    $\ev_E(p)=\sum_{r\in \left[0,1\right]} r\cdot p(r)$ if
    $p\in \dist X$ ($0$ otherwise). Furthermore,
    $\ev_*\colon \dist [0,1]+1\to [0,1]$ with $\ev_*(p)=1$ if
    $p=\checkmark$ ($0$ otherwise). We extend this to a logic by
    defining $\cl'$ (Assumption~\ref{assum:PropositionalClosure}): we add as
    operators the constant~$1$, $\min(\varphi,\varphi')$, $1-\varphi$
    and $\varphi\dotminus q$ for a rational $q$ (where $\varphi$ is a
    formula), as in similar logics for probabilistic transition
    systems~\cite{bw:behavioural-pseudometric}.
  \end{enumerate}
\end{example}

To ensure that the requirements of \cref{sec:adjoint-logic} are met,
we have to show compatibility of the logic function. To this end, we
introduce the notion of compability of $\cl'$.

\begin{definition}
  Given a closure
  $\cl'_X \colon \power \Ccat (X,\tobject) \to \power
  \Ccat(X,\tobject)$, we say that $\cl'_X$ is \emph{compatible} %
  if the map $\Lambda_X \circ\ \cl'_X$ (for each $X\in\Ccat$) is
  compatible with the closure $\cl_X=\gamma_X \circ \alpha_X$ induced
  by the adjoint situation in \cref{thm:fibredadjoint}, i.e.,
  $
    \Lambda_X \circ \cl_X' \circ \cl_X \subseteq \cl_{FX} \circ
    \Lambda_X \circ \cl_X'.
  $
\end{definition}
The following results hold under
Assumptions~\ref{assum:ConfExists}-\ref{assum:PropositionalClosure}
and thus, we avoid stating them in various lemma/theorem statements.
\begin{propositionrep}
  \label{prop:compatibility}
  For a given compatible closure $\cl'_X$, the above logic function
  $\lo_X$ %
  is $\cl_X$-compatible, i.e.,
  $\lo_X \circ \cl_X\subseteq \cl_X \circ \lo_X$. %
\end{propositionrep}

\begin{proof}
  We use \cref{prop:compositionality} to prove compatibility of $\lo$
  in a compositional manner. Let
  $\lo_1=\power (c^\bullet) \circ \Lambda_X \circ \cl'_X$ and
  $\lo_2=\Theta_X$, i.e., $\lo_X=\lo_1 \cup \lo_2$.

  \noindent
  First, we derive
  \begin{align*}
    \lo_1 \circ \cl_X =&\ \power (c^\bullet) \circ \Lambda_X \circ \cl'_X \circ \cl_X && (\because\ \text{$\cl_X'$ is compatible})\\
    \subseteq &\ \power (c^\bullet) \circ \cl_{FX} \circ \Lambda_X \circ \cl_X' \\
    =&\ \power (c^\bullet) \circ \gamma_{FX} \circ \alpha_{FX} \circ \Lambda_X \circ \cl_X'  && (\because\ \text{lax naturality of $\gamma$})\\
    \subseteq &\ \gamma_{X} \circ c^* \circ \alpha_{FX} \circ \Lambda_X \circ \cl_X' && (\because\ \text{naturality of $\alpha$})\\
    =&\ \gamma_X \circ \alpha_X \circ \power(c^\bullet)\circ \Lambda_X \circ \cl'_X \\
    =&\ \cl_X \circ \lo_1 .
  \end{align*}
  \noindent
  Note that $\lo_2$ is a constant function (thus $\cl$-compatible) and
  thus \cref{prop:compositionality} is applicable.
\end{proof}

We will now study equivalent conditions and special cases in which
compatibility holds. First, it is easy to see that $\cl' = \cl$ is
always compatible, but typically introduces infinitary
operators. Moreover, if~$\cl'$ is the identity (that is, there are no
propositional operators), then compatibility of $\cl'$ reduces to
$\cl$-compatibility of $\Lambda$.

\begin{lemmarep}
  \label{lem:compatibility-cond}
  Let $\cl_X'$ be a subclosure of $\cl_X$. It holds that $\cl_X'$ is
  compatible if and only if
  $ \alpha_{FX} \circ \Lambda_X \circ \cl_X' \preceq \alpha_{FX}
    \circ \Lambda_X \circ \cl_X. $
\end{lemmarep}

\begin{proof}
  Since we assume that $\cl'_X$ is a subclosure of $\cl_X$,
  compatibility of $\cl'$ simplifies to
  \[ \Lambda\circ\cl_X \subseteq \cl_{FX}\circ\Lambda_X\circ\cl'_X =
    \gamma_{FX}\circ\alpha_{FX}\circ\Lambda_X\circ\cl'_X. \]
  Due to the properties of a Galois connection this is equivalent to
  \[
    \alpha_{FX}\circ\Lambda_X\circ\cl'_X\preceq \alpha_{FX}\circ
    \Lambda_X\circ\cl_X. \qedhere \]
\end{proof}

We next adapt the separation property establishing expressiveness of
\emph{graded logics} w.r.t.\ \emph{graded
  semantics}~\cite{DorschEA19,fswbgkm:graded-logics-em} to the present
setting, an additional twist being that the conformance w.r.t.\ which
modalities must be separating is the one induced by the modalities
themselves.
\begin{definition}[Depth-1 self-separation]\label{def:self-separation}
  A set $S\subseteq \Ccat(X,\tobject)$ of predicates is \emph{initial}
  for $d\in\Phi X$ if $\alpha_X(S) = d$. Let $\cl_X'$ be a subclosure
  of $\cl_X$. The \emph{depth-1 self-separation} property holds if for
  every $S$ that is closed under $\cl'_X$ (i.e., $S = \cl'_X(S)$) and
  initial for $d$, it follows that $\Lambda_X(S)$ is initial for
  $\kant_X(d)$ %
  where $\kant_X = \alpha_{FX}\circ \Lambda_X\circ \gamma_X$.
\end{definition}

\begin{lemmarep}
  \label{lem:compatibility-depth1}
  Let $\cl_X'$ be a subclosure of $\cl_X$. Then $\cl_X'$ is
  compatible if and only if the depth-1 self-separation property holds.
\end{lemmarep}

\begin{proof}
  We first show that the depth-1 self-separation property implies
  compatibility of $\cl'_X$, by proving the inequality of
  \cref{lem:compatibility-cond}. So let
  $T\subseteq \Ccat(X,\tobject)$. Define $S = \cl'_X(T)$ and observe
  that $S$ is closed under $\cl'$. Put $d = \alpha_X(S)$; by
  definition,~$S$ is initial for $d$. By the depth-1 self-separation
  property, it follows that $\Lambda_X(S)$ is initial for $\kant_X(d)
  = \alpha_{FX}(\Lambda_X(\gamma_X(d)))$. Hence
  \begin{align*}
    \alpha_{FX}(\Lambda_X(\cl'_X(T))) &= \alpha_{FX}(\Lambda_X(S)) =
    \alpha_{FX}(\Lambda_X(\gamma_X(d))) =
    \alpha_{FX}(\Lambda_X(\gamma_X(\alpha_X(S)))) \\
    &=
    \alpha_{FX}(\Lambda_X(\cl_X(\cl'_X(T)))) =
    \alpha_{FX}(\Lambda_X(\cl_X(T))),
  \end{align*}
  using the fact that $\cl'$ is a subclosure of $\cl$. From
  \cref{lem:compatibility-cond} we infer compatibility of
  $\cl'_X$.

  We now show the other direction, i.e., compatibility of $\cl'$
  implies depth-1 self-separation.  Let $S\subseteq \Ccat(X,\tobject)$ be a
  set that is closed under $\cl'_X$ and $S$ is initial for $d$, i.e.,
  $\alpha_X(d) = S$.  We have to show that $\Lambda_X(S)$ is initial
  for $\kant_X(d)$, i.e., that
  $\alpha_{FX}(\Lambda_X(S)) =
  \alpha_{FX}(\Lambda_X(\gamma_X(d)))$. The inequality
  $\alpha_{FX}(\Lambda_X(S)) \succeq
  \alpha_{FX}(\Lambda_X(\gamma_X(d)))$ follows trivially from the
  properties of the (contra-variant) Galois connection, since
  $S\subseteq \gamma_X(\alpha_X(S)) = \gamma_X(d)$.  For the other
  direction we use \cref{lem:compatibility-cond}, which gives us:
  \begin{align*}
    \alpha_{FX}(\Lambda_X(S)) &= \alpha_{FX}(\Lambda_X(\cl'_X(S)))
    \preceq \alpha_{FX}(\Lambda_X(\cl_X(S)))  \\
    &= \alpha_{FX}(\Lambda_X(\gamma_X(\alpha_X(S)))) =
    \alpha_{FX}(\Lambda_X(\gamma_X(d))).\qedhere
  \end{align*}
\end{proof}

Finally, we study a sufficient condition on evaluation maps
ensuring $\cl$-compatibility of $\Lambda$.

\begin{lemmarep}
  \label{lem:LambdaCompatibleForNaturalEvs}
  If each evaluation map $\ev_\lambda$ arises as a natural
  transformation $\eta\colon F \Rightarrow \text{Id}$ or
  $\eta\colon F \Rightarrow \mathbf\tobject$ ($\mathbf\tobject$ is the
  constant functor mapping every object to $\tobject$), that is
  $\ev_\lambda = \eta_\tobject$, then $\Lambda$ is compatible with
  $\cl$.
\end{lemmarep}

\begin{proof}
  We denote the components of the natural transformation by
  $\ev_{\lambda,X}$, hence the evaluation map itself correspondgs to
  $\ev_{\lambda,\tobject}$.

  Let $X\in\Ccat$ and $S\subseteq \Ccat(X,\tobject)$. We need to show that $\Lambda_X\cl_X(S) \subseteq \cl_{FX} \circ \Lambda_X (S)$.

  So assume $\ev_\lambda \circ Fk \in \Lambda_X \circ \cl_X (S)$ for some $\lambda\in\Lambda$ and $k\in \cl_X (S)$ (i.e., $\alpha_X (S) \preceq k^*(d_\tobject)$). Now we do the case distinction based on the type of $\ev_\lambda$.

  \begin{itemize}
    \item Let $\ev_\lambda \colon F \Rightarrow \mathbf{\tobject}$. Then we find
        \begin{align*}
           \alpha_{FX}(\Lambda_X(S)) =&\ \bigwedge_{\lambda'\in \Lambda,h\in S} (\ev_{\lambda',\tobject} \circ Fh)^* (d_\tobject) && (\because\ \ev_{\lambda'}\ \text{is natural})\\
           =&\ \bigwedge_{\lambda'\in \Lambda,h\in S} \ev_{\lambda',X}^* (d_\tobject)\\
           =&\ \bigwedge_{\lambda'\in \Lambda} \ev_{\lambda',X}^* (d_\tobject)\\
           \preceq &\ \ev_{\lambda,X}^* (d_\tobject)
        \end{align*}
        Thus, $\ev_{\lambda,X} = \ev_{\lambda,\tobject} \circ Fk \in \cl_{FX}(\Lambda_X(S))$.
    \item Let $\ev_\lambda \colon F \Rightarrow \text{Id}$. Using $k\in\cl_X(S)$ we have $\alpha_X(S) \preceq k^* (d_\tobject)$ and since reindexing functor are order preserving, so we find that
        \begin{align*}
        \ev_{\lambda,X}^* & (\alpha_X(S)) \preceq  \ev_{\lambda,X}^* k^* (d_\tobject)\\
          \iff &\ \ev_{\lambda,X}^* \left(\bigwedge_{h\in S} h^*(d_\tobject) \right)\preceq  \ev_{\lambda,X}^* k^* (d_\tobject)\\
          \iff & \bigwedge_{h\in S} \left(\ev_{\lambda,X}^* h^*(d_\tobject) \right) \preceq \ev_{\lambda,X}^* k^* (d_\tobject) \\
          \iff & \bigwedge_{h\in S} (h\circ \ev_{\lambda,X})^* (d_\tobject) \preceq (k\circ \ev_{\lambda,X})^* (d_\tobject)\\
          \iff & \bigwedge_{h\in S} (\ev_{\lambda,\tobject}\circ Fh)^*(d_\tobject) \preceq (\ev_{\lambda,\tobject}\circ Fk)^*(d_\tobject)
        \end{align*}
        Moreover $\alpha_{FX}\Lambda_X(S) =  \bigwedge_{\lambda'\in \Lambda,h\in S} (\ev_{\lambda',\tobject} \circ Fh)^* (d_\tobject) \preceq  \bigwedge_{h\in S} (\ev_{\lambda,\tobject} \circ Fh)^* (d_\tobject)$ and thus using the above (derived) inequality we find
        \[
        \alpha_{FX}\Lambda_X(S) \preceq \bigwedge_{h\in S} (\ev_{\lambda,\tobject} \circ Fh)^* (d_\tobject) \preceq (\ev_{\lambda,\tobject}\circ Fk)^*(d_\tobject)
        \]
        Thus,
        $\ev_{\lambda,\tobject}\circ Fk \in
        \cl_{FX}(\Lambda_X(S))$. \qedhere
      \end{itemize}
\end{proof}

\begin{example}
  \label{ex:compatibility}
  We establish compatibility for the logics considered in
  \cref{ex:evaluation-maps-logic}. In branching-time logics in
  general, depth-1 self-separation usually boils down to establishing
  a Stone-Weierstraß type property saying that if
  $S\subseteq \Ccat(X,\tobject)$ is initial and closed under $\cl'_X$,
  then~$S$ is dense in $\Ccat(X,\tobject)$, for suitably
  restricted~$X$~\cite{Forster_et_al:CSL.2023:Density}. For finitary
  set functors such as~$\power_\mathit{fin}$ or $\dist$, it suffices
  to prove self-separation on finite~$X$. Additional details are as
  follows.
  \begin{enumerate}[wide]
  \item In the case of unlabelled transition systems, %
    we are given an equivalence relation~$R$ on a finite set~$X$, a
    set~$S\subseteq\set(X,2)$ that is initial for~$R$ and closed under
    Boolean combinations, and $A,B\in\power_{\mathit{fin}}(X)$ that are
    distinguished by some predicate
    $\Diamond_X f\colon\power_{\mathit{fin}}(X)\to 2$ where~$f$ is
    invariant under~$R$. We then have to show that~$A,B$ are
    distinguished by $\Diamond g$ for some $g\in S$. But by functional
    completeness of Boolean logic and because~$X$ is finite,~$S$ is in
    fact the set of \emph{all} $R$-invariant functions $X\to 2$, so we
    can just take~$g=f$.
  \item The argument is similar for probabilistic transition systems,
    with some additional considerations necessitated by the
    quantitative setting. We are now given a
    set~$S\subseteq\set(X,[0,1])$ that is initial for~$d$ and closed
    under propositional operators as per
    \cref{ex:evaluation-maps-logic}.\ref{item:prob}. By a variant of
    the Stone-Weierstraß theorem, this implies that~$S$ is dense in
    the space of non-expansive maps $(X,d)\to[0,1]$
    (see~\cite{WildEA18}), which means that in an argument as in the
    previous item, we can take~$g$ to range over functions in~$S$ that
    approximate the given non-expansive
    function~$f\colon (X,d)\to[0,1]$ arbitrarily closely, using
    additionally that the predicate lifting induced by~$\ev_E$ as in
    \cref{ex:evaluation-maps-logic}(\ref{item:prob}) is
    non-expansive~\cite{WildSchroder22}.
  \end{enumerate}
\end{example}

\subsection{Towards a Generic Behaviour Function}
\label{subsec:gen-behaviour}

Building on %
the previous section, we define the behaviour function
$\be_X\colon \Phi X \to \Phi X$ as
$\be_X = \alpha_X \circ \lo_X \circ \gamma_X$ and -- under the
assumption of compatibility -- we have\footnote{Note that the
  adjunction defined in \cref{thm:fibredadjoint} is
  contravariant. Hence the least fixpoint of $\be$ from
  \cref{thm:fixpoint-preservation} becomes the greatest fixpoint
  $\nu \be_X$ wrt.\ the lattice order $\preceq$.}
$\alpha_X(\mu \lo_X) = \nu \be_X$. In other words, the notions of logical and
behavioural conformances coincide.

This motivates a closer investigation of $\nu\be_X$: in what sense
does it coincide with known behavioural equivalences or metrics?
Defining a behavioural conformance (that is, an element of $\Phi X$)
in a fibrational setting is typically done by taking the greatest
fixpoint of a function defined in two steps: the lifting of a
conformance $\Phi X$ to $\Phi(FX)$, followed by a reindexing
via~$c$.  Here, we consider
Kantorovich-style~\cite{bbkk:coalgebraic-behavioral-metrics} or
codensity~\cite{CodensityGames} liftings based on the evaluation
maps. Kantorovich liftings have originally been used to lift metrics
on a set~$X$ to metrics of probability distributions over $X$. In the
probabilistic case, an alternative characterization is given via
optimal transport plans in transportation theory (earth mover's
distance) \cite{v:optimal-transport}.

We use the natural transformation $\Lambda$ introduced earlier and
consider the composite
$\kant_X =\alpha_{FX}\circ \Lambda_X \circ \gamma_X$, the mentioned
Kantorovich lifting. %
If $F=\dist$ and the evaluation map is expectation, then we
obtain exactly the classical Kantorovich lifting.

Given a coalgebra $c\colon X \to FX$, we use the \emph{behaviour
  function} $\be_X = c^* \circ \kant_X \land\ \alpha_X(\Theta_X)$;
here, $\Theta_X$ is a set of constants as in \cref{subsec:gen-logic}
(Asumption~\ref{assum:ConstantsExist}).

\begin{example}
  We derive the %
  behaviour functions for %
  Examples~\ref{ex:evaluation-maps-logic} and~\ref{ex:compatibility}.
  \begin{enumerate}[wide]
  \item In the case $F=\power$ and $\V=2$, the lifting
    $\kant_X(R) \subseteq \power X \times\power X$ of an equivalence
    relation $R\subseteq X \times X$ is the Egli-Milner lifting, i.e.\
    $ U \mathrel{\kant_X(R)} V \iff \forall_{x\in U}\exists_{y\in V}
    x\mathrel R y \land \forall_{y\in V} \exists_{x\in U}\ x\mathrel R
    y.  $ It is well-known that the greatest fixpoint of
    $\be_X=c^*\circ\kant_X$ is precisely Park-Milner bisimilarity.
  \item In the case $FX=\dist X+1$ and $\V=\left[0,1\right]$, we
    obtain the lifting $\kant_X(d)\in\pmet(\dist X+1)$ of a
    pseudometric $d\in\pmet(X)$. It is easy to see that
    $\kant_X(d)(p_1,p_2)$ is the distance given by the classical
    Kantorovich lifting of $d$ if $p_1,p_2\in\mathcal{D}X$. If
    $p_1=\checkmark=p_2$, then the distance is~$0$, otherwise~$1$.
    The least fixpoint (under the usual order on $\left[0,1\right]$)
    of the behaviour function $\be_X=c^*\circ\kant_X$ agrees with
    standard notions of bisimulation distance
    (e.g.~\cite{bw:behavioural-pseudometric}).
  \end{enumerate}
\end{example}

\noindent We conclude the section by showing that behaviour functions
defined in this way are actually the ones obtained from the logic
function.  For the diagram underlying the proof see
\cref{fig:AdjointSetup}.

\begin{figure}
  \centering
\begin{tikzcd}
  \power(\Ccat(X,\tobject))
    \arrow[rr,"\alpha_X",yshift=1ex,"\perp"']
    \arrow[d,"\Lambda_X"]
    \arrow[loop, out=105, in = 50, looseness=4, "\cl'_X",dashed]
    \arrow[loop,dashed,out=155,in=125,looseness=4,"\Theta_X"]
    \arrow[dd,rounded corners, dashed, to path={--([xshift=-2ex]\tikztostart.west) node[left,yshift=-7ex]{$\scriptstyle\lo_X$}
    |- (\tikztotarget.west)}]&&
  \Phi(X)^{\text{op}}
    \arrow[ll,"\gamma_{X}",yshift=-1ex]
    \arrow[d,"\kant_X",dashed]
    \arrow[dd,rounded corners, dashed, to path={--([xshift=2ex]\tikztostart.east) node[right,yshift=-7ex]{$\scriptstyle\be_X$}
    |- (\tikztotarget.east)}]\\
   \power(\Ccat(FX,\tobject))
    \arrow[rr,"\alpha_{FX}",yshift=1ex,"\perp"']
    \arrow[d,"\power(c^\bullet)"]&&
  \Phi(FX)^{\text{op}}
    \arrow[ll,"\gamma_{FX}",yshift=-1ex]
    \arrow[d,"c^*"]\\
   \power(\Ccat(X,\tobject))
    \arrow[rr,"\alpha_X",yshift=1ex,"\perp"']&&
  \Phi(X)^{\text{op}}
    \arrow[ll,"\gamma_{X}",yshift=-1ex]
\end{tikzcd}
  \caption{The adjoint setup with $\scriptstyle\lo_X= \power(c^\bullet)\circ \Lambda_X \circ \cl'_X \cup\ \Theta_X$ and $\scriptstyle\be_X= c^* \circ \kant_X \cup\ \alpha_X(\Theta_X)$.}\label{fig:AdjointSetup}
\end{figure}

\begin{theoremrep}
  \label{thm:genHML}
  Assume that $\cl'_X$ is a subclosure of $\cl_X$ and compatible. For
  a set\/~$\Theta_X$ of constants and a coalgebra $c\colon X\to FX$,
  the logic function
  $ \lo_X=\power (c^\bullet) \circ \Lambda_X \circ \cl'_X \cup \
  \Theta_X$ induces
  $\be_X = c^*\circ \kant_X\land\ \alpha_X(\Theta_X)$, i.e.,
  $\alpha_X \circ \lo_X \circ \gamma_X = \be_X$.
\end{theoremrep}

\begin{proof}
  Define $\lo_1$, $\lo_2$ as in the proof of
  \cref{prop:compatibility}. Then we have:
  \begin{align*}
    \alpha_X \circ \lo_X \circ \gamma_X =&\ \alpha_X \circ \lo_1 \circ \gamma_X \land \alpha_X \circ \lo_2 \circ \gamma_X \\
    =&\ \alpha_X \circ \power (c^\bullet) \circ \Lambda_X \circ \cl'_X \circ \gamma_X \land\ \alpha_X \circ \Theta_X \circ \gamma_X \\
    & \quad (\because \cl'\circ \cl =\cl, \text{\cref{lem:char-subclosure}})\\
    =&\ \alpha_X \circ \power (c^\bullet) \circ \Lambda_X \circ \gamma_X \land\ \alpha_X(\Theta_X) \\
    & \quad (\because \text{$\alpha$ is natural})\\
    =&\ c^*\circ \alpha_{FX} \circ \Lambda_X \circ \gamma_X \land\
    \alpha_X(\Theta_X) \\
    =&\ \be_X. && \qedhere
  \end{align*}
\end{proof}
\noindent Putting everything together via
\cref{thm:fixpoint-preservation}, if $\cl'_X$ is a compatible closure,
then we have $\alpha_X(\mu \lo_X) = \nu \be_X$, that is, logical
conformance coincides with behavioural conformance.

\section{Logics for Quantale-valued Simulation Distances}
\label{sec:quantale-branching-logic}

We next consider a quantitative modal logic~$\qlog$ that we show to be
expressive for similarity distance (a behavioural directed metric)
under certain conditions. Throughout this section, our working
category $\Ccat$ is $\set$, we have a fixed functor $F$ on $\set$, and
a fixed quantale $\tobject=\V$ with distance $d_\V = [\_,\_]$ being
the internal hom. Furthermore $\Phi X = \dpmet_\V(X)$. In this section
we assume naturality of $\gamma$, which holds for the quantales given
in \cref{ex:quantale}.
\[
  \varphi\in \qlog ::= \bigwedge\nolimits_{i\in I} \varphi_i \ \mid\
  \varphi \otimes v \ \mid \ d_\V(v,\varphi)\ \mid\ [\lambda]\varphi
  \quad (\text{for $v\in\V,\lambda\in\Lambda,I\in\set$})
\]
This logic is the positive fragment of quantale-valued coalgebraic
modal logic~\cite{WildSchroder21,Forster_et_al:CSL.2023:Density}, and
generalizes logics for real-valued simulation
distance~\cite{WildSchroder22} to the quantalic setting. The first
three operators are regarded as propositional operators of $\cl'$,
while the~$[\lambda]$ are the modalities. We do not use explicit
constants ($\Theta_X = \emptyset$), but note that constant
truth~$\top$ is included as the empty meet. On a coalgebra
$c\colon X \to FX \in\set$, we interpret a formula $\varphi$ as a
function $\sem{\varphi}\colon X \to \V$ as usual (by structural
induction on terms). Note that we do not allow negation, as we aim to
characterize similarity distance. Disjunction
could be included but, as in the two-valued
case~\cite{g:linear-branching-time}, is not needed to characterize
simulation. %
Meet is infinitary, so the logic function $\lo$ does not reach its
least fixpoint in $\omega$ steps.

The next results show that the three operators (as well as join) are
all non-expansive and hence $\cl'$ is a subclosure of $\cl$
(cf.~\cref{lem:charClosure}), which moreover is compatible.

\begin{propositionrep}
  \label{prop:directed-gamma-closure}
  Infinitary meets, infinitary joins, negative scaling $d_\V(v,\_)$,
  and positive scaling $\_ \otimes v$ (for $v\in \V$) are
  non-expansive.
\end{propositionrep}

\begin{proof}
  For the non-expansiveness of arbitrary meets, let $v_i,w_i\in\V$
  ($i\in I$). We need to show that
  $d_\V(v_i,w_i) \leq d_\V(\bigwedge_{i\in I}
  v_i,\bigwedge_{i\in I} w_i)$. By reflexivity we have
  $d_\V(v_i,w_i) \leq d_\V(v_i,w_i)$ for each $i\in
  I$. Then, by $v\otimes\_ \dashv d_\V(v,\_)$, we have
  $d_\V(v_i,w_i)\otimes v_i \leq w_i$ for each $i$. Thus, taking
  infima both sides and since infimum is order preserving we get
  \[
  \bigwedge_{i\in I} (d_\V(v_i,w_i) \otimes v_i) \leq \bigwedge_{i\in I} w_i.
  \]
  Moreover since $v\otimes\_$ is order preserving we have
  \[
    d_\V(v_i,w_i) \otimes \bigwedge_{i\in I} v_i \leq \bigwedge_{i\in I}
    (d_\V(v_i,w_i) \otimes v_i).
  \]
  So by transitivity of $\leq$ we get
  $ d_\V(v_i,w_i) \otimes \bigwedge_{i\in I} v_i \leq
  \bigwedge_{i\in I} w_i$, which by
  $v\otimes\_ \dashv d_\V(v,\_)$ is equivalent to
  $d_\V(v_i,w_i) \leq d_\V(\bigwedge_{i\in I}
  v_i,\bigwedge_{i\in I} w_i)$.

  For the non-expansiveness of arbitrary joins, we prove, using
  distributivity of $\otimes$ over join:
  \begin{align*}
    & \forall_{i\in I}\
    d_\V(v_i,w_i) \leq d_\V(v_i,w_i)\\
    \iff &\ \forall_{i\in I}\ d_\V(v_i,w_i) \otimes v_i \leq w_i\\
    \implies &\ \bigvee_{i\in I} (d_\V(v_i,w_i) \otimes v_i) \leq \bigvee_{i\in I} w_i\\
    \iff &\ d_\V(v_i,w_i) \otimes \bigvee_{i\in I} v_i \leq \bigvee_{i\in I} w_i\\
    \iff &\ d_\V(v_i,w_i) \leq d_\V(\bigvee_{i\in I} v_i,\bigvee_{i\in
      I} w_i).
  \end{align*}

  For non-expansiveness of negative scaling, observe that
  $d_\V(v_1,v_2) \leq d_\V(d_\V(v,v_1),d_\V(v,v_2))$ by the triangle
  inequality. From commutativity of $\otimes$ it follows that
  $d_\V(v_1,v_2) \otimes d_\V(v,v_1) \leq d_\V(v,v_2)$. In the last
  step we use the adjunction property to obtain:
  $d_\V(v,v_1) \otimes d_\V(v_1,v_2) \leq d_\V(v,v_2)$.

  For non-expansiveness of positive scaling, we again start from
  $d_\V(v_1,v_2)\le d_\V(v_1,v_2)$ (reflexivity). Then, using the
  adjunction property, we obtain $d_\V(v_1,v_2)\otimes v_1\le v_2$.
  This implies that
  $d_\V(v_1,v_2)\otimes v_1\otimes v\le v_2\otimes v$.  And again
  using the adjunction property we can conclude that
  $d_V(v_1,v_2)\le d_\V(v_1\otimes v,v_2\otimes v)$.
\end{proof}

\begin{toappendix}
  In order to show compatibility of $\cl_X'$, we need the following
  result that rewrites a predicate $p\in \cl_X$ in a certain normal
  form.

\begin{proposition}
  Let $S \subseteq \set(X,\V)$. Then every non-expansive predicate
  $p\in\cl_X (S)$ can be rewritten as $\bigvee_{x\in X} p_x$ where
  $p_x(y)=\big(\bigwedge_{h\in S} d_\V(h(x),h(y)) \big)\otimes p(x)$.
\end{proposition}

\begin{proof}
  Remember that we assume $\gamma_X$ being natural.
  I.e., for any $p\in\cl_X(S)$ there is a non-expansive function
  $op\colon(\V^S,d_\V^S) \to (\V,d_\V)$ such that
  $p=op\circ\langle S\rangle$ (cf.~\cref{lem:charClosure}). In
  particular, the function $op$ is given by the (enriched) left Kan
  extension $\text{Lan}_{\langle S \rangle} p$ of $p$ along
  $\langle S\rangle$. Since $\V$ is cocomplete and has tensors (given
  by $\otimes$), thus using the coend description of left Kan
  extension:
  \[
    \text{Lan}_{\langle S \rangle} p (u) = \bigvee_{x\in X}
    d_\V^S(\langle S \rangle x,u) \otimes p(x) = \bigvee_{x\in X}
    \left(\bigwedge_{h\in S} d_\V(h(x),u(h))\right)\otimes p(x),
  \]
  where $u\colon S\to\V$.

  In particular, for $y\in X$ we have
  \[
    p(y) = \text{Lan}_{\langle S \rangle} p (\langle S \rangle y) =
    \bigvee_{x\in X} (\bigwedge_{h\in S} d_\V(h(x),h(y))) \otimes p(x)
    = \bigvee_{x\in X} p_x(y).\qedhere
  \]
\end{proof}
\end{toappendix}
\begin{propositionrep}
  If %
  each $\lambda\in \Lambda$ is sup-preserving (i.e.
  $\lambda_X (\bigvee P) = \bigvee \power(\lambda_X)(P)$, for every
  subset $P \subseteq \set(X,\V)$,)
  then the sub-closure $\cl'$ of $\cl$ as above is compatible.
\end{propositionrep}

\begin{proof}
  Since $\cl'$ is a subclosure of $\cl$, it suffices to show that
  $\Lambda_X(\cl_X (S)) \subseteq \cl_{FX}(\Lambda_X(\cl_X'(S)))$ (for
  $S\subseteq\set(X,\V)$). Assume
  $\lambda_X(p) \in \Lambda_X(\cl_X(S))$, for some
  $\lambda\in\Lambda,p\in\cl_X(S)$. Using the previous
  proposition %
  we find
  \[
  \lambda_X(p) = \lambda_X(\bigvee_{x\in X} p_{x})
    = \bigvee_{x\in X} \lambda_X (p_x)
  \]
  Clearly, $p_{x} \in \cl_X'(S)$ because
  $d_\V(h(x),h(\_)) \in \cl_X'(S)$ (for $h\in S$) and $\cl_X'$ is closed
  under arbitrary meets and positive scaling. Thus $\lambda_X(p_x)\in \Lambda_X\cl_X'(S)$ (for each $x\in X$); whence $\lambda_X(p_x) \in \cl_{FX}\Lambda_X \cl_X'(S)$. And since full closure (in particular $\cl_{FX}$) is closed under arbitrary joins, we conclude that $\lambda_X(p) \in \cl_{FX}\Lambda_X \cl_X'(S)$.
\end{proof}

\noindent Diamond-like modalities (for powerset or fuzzy powerset) are
typically sup-preserving. In such cases, \cref{thm:genHML} yields
expressiveness of the above logic for similarity distance, defined as
the greatest fixpoint of $\be_X(d) = c^*\circ\kant_X(d)$ where
$\kant_X$ is the directed Kantorovich lifting.

\section{The Adjoint Setup in an Eilenberg-Moore Category}
\label{sec:adjointlogicEM}

While we have seen in the examples of the previous sections that the
framework can be instantiated to coalgebras living in $\set$, thus
providing Hennessy-Milner theorems for bisimilarity, we are now
interested in tackling trace equivalences and trace metrics. To this
end, %
we  work in Eilenberg-Moore categories
\cite{sbbr:powerset-coalgebraically}, which also allows us to
determinize a coalgebra using
the generalized powerset construction
(cf. \ref{subsec:coalgebras-EM}).

In particular, we instantiate the adjoint setup to the category $\alg
T$ of $T$-algebras (for some monad $T$ on $\Ccat$), provide
conditions guaranteeing compatibility, and characterize the behaviour
function. Furthermore, taking inspiration from
\cite{Schroeder2007:ExpressivityCoalgModalLog}, we also introduce a
general syntax for modal formulas that can be interpreted
over coalgebras living in $\alg T$.  As introduced in
\cref{subsec:coalgebras-EM}, we fix a coalgebra $c\colon X\to FTX$
living in $\Ccat$ and its determinization $c^\#\colon LX \to \tilde F
LX$ in $\alg{T}$ via a distributive law $\zeta\colon TF\Rightarrow
FT$.

We assume
an indexed category $\Psi\colon \Ccat^{\text{op}} \to \pos$ (that has
fibred limits) and lift it to the category $\alg T$ of $T$-algebras
by postcomposition, that is, $\Phi = \Psi\circ R$ (thus ensuring \ref{assum:ConfExists}):
\[ \alg{T}^{\text{op}} \xrightarrow{R} \Ccat^{\text{op}}
\xrightarrow{\Psi} \pos.
\]
Here,~$R$ is the forgetful functor in the free-forgetful adjunction
$L\dashv R\colon \Ccat \to \alg{T}$ from
\cref{subsec:coalgebras-EM}. To handle \ref{assum:TruthConfExists}, we fix a truth value object $\tobject\in\Ccat$ equipped with a $T$-algebra structure $o\colon T\tobject \to \tobject$ and $d_\tobject \in \Phi \tobject$.
These assumptions ensure that
\cref{thm:fibredadjoint} becomes applicable.  We will denote the
reindexing for $\Phi$ by $\_^*$, while we overload the notation and
specify the reindexing in both $\Ccat$ and $\alg{T}$ by $\_^\bullet$.

We focus on free algebras $LX = (TX,\mu_X)$ (over $X\in\Ccat$) and
apply \cref{thm:fibredadjoint} to the above-mentioned indexed category
$\Psi \circ R$, which gives the adjoint situations depicted by shaded
rectangles in \cref{fig:AdjointSetupAlgebras}.
We note that the middle hom-set $\alg{T}(LX,(\tobject,o))$ is
isomorphic to $\Ccat(X,\tobject)$ at the top left -- due to the
free-forgetful adjunction -- with the bijection between the respective
powersets witnessed by $\alpha',\gamma'$.
This allows us to define the logic function on the lattice
$\power(\Ccat(X,\tobject))$, which is a simpler structure than
$\power(\alg{T}(LX,(\tobject,o)))$. In particular, formulas can then
be evaluated directly on the state space $X$.

\subsection{Logic and Behaviour Function for Coalgebras in Eilenberg-Moore}
\label{subsec:logic-behaviour-algT}

Recall from \cref{subsec:gen-logic} that we need evaluation maps in
the working category to define a logic function. So, to ensure \ref{assum:EvalsExist}, we assume a set
$\Lambda$ of evaluation maps for $\tilde F$, i.e., a family
$(\tilde F(\tobject,o) \xrightarrow{\ev_\lambda}
(\tobject,o)\in\alg{T})_{\lambda\in\Lambda}$ of algebra
homomorphisms. More concretely, a $\Ccat$-arrow
$\ev_\lambda\colon F\tobject\to \tobject$ is such an algebra
homomorphism if it satisfies
$o\circ T\ev_\lambda = \ev_\lambda\circ Fo\circ\zeta_\tobject$.

As in \cref{sec:adjointlogic}, this %
induces a natural transformation $\Lambda$. Since every homomorphism
is also a map in $\Ccat$, we can define $\Lambda'_X$, the predicate
lifting on $\power \Ccat(X,\tobject)$:
\[
  \power \Ccat(X,\tobject) \xrightarrow{\alpha'_X}
  \power\alg{T}(LX,(\tobject,o)) \xrightarrow{\Lambda_{LX}} \power
  \alg{T}(\tilde FLX,(\tobject,o)) \xrightarrow{\power(R)} %
  \power \Ccat(FTX,\tobject).
\]
Note that $\Lambda'$ is a natural transformation
(since $\alpha,\Lambda$ are natural transformations and $R$ is a
functor); on components it can be easily characterized as
follows: %
\begin{lemmarep}
  \label{lem:lambda-prime}
  We have that
  $\Lambda'_X(S) = \{\ev_\lambda \circ Fo\circ FTh\mid h\in S\}$ where
  $S\subseteq \Ccat(X,\tobject)$.
\end{lemmarep}

\begin{proof}
  Given $S$, $\alpha'_X$ maps every $(h\colon X\to \tobject) \in S$ to
  $o\circ Lh = o\circ Th\colon LX\to \tobject$. (Note that $o$ equals
  $\epsilon_{(\tobject,o)}$, the component for $(\tobject,o)$ of the
  co-unit.)  $\Lambda_{LX}$ maps this arrow to
  $\ev_\lambda\circ\tilde{F}(o\circ Lh) = \ev_\lambda\circ Fo\circ
  FTh$.  %
\end{proof}

That is, $\Lambda'_X(S)$ is obtained by first lifting the predicates
from $\Ccat(X,\tobject)$ to $\Ccat(TX,\tobject)$ via the evaluation
map $o\colon T\tobject\to\tobject$ and then to $\Ccat(FTX,\tobject)$ via
$\ev_\lambda\colon F\tobject\to\tobject$. This process can be seen as
applying a ``double modality'' for $T$ and $F$.

\begin{figure}
  \centering
    \begin{tikzcd}[execute at end picture=
    {\fill[gray,opacity=0.2] (A.north west) -- (B.north east) -- (B.south east) -- (A.south west);
    \fill[gray,opacity=0.2] (C.north west) -- (D.north east) -- (D.south east) -- (C.south west);}
    ]
    \power \Ccat(X,\tobject)
    \arrow[loop above,"\Theta'_X"]
    \arrow[r,"\alpha'_X",yshift=1ex,"\cong"']
    \arrow[dd,"\Lambda_X'" description]
    \arrow[ddd,rounded corners, dashed, to path={--([xshift=-2ex]\tikztostart.west) node[left,yshift=-12ex]{$\scriptstyle\lo'_X$}
    |- (\tikztotarget.west)}]
     & |[alias=A]|
     \power (\alg{T}(LX,(\tobject,o))) \arrow[l,"\gamma'_X",yshift=-1ex]  \arrow[r,"\alpha_{LX}",
     yshift=1ex,"\perp"'] \arrow[d,"\Lambda_{LX}"] \arrow[loop above,"\Theta_{LX}"]
      & |[alias=B]|
      \Psi^{\text{op}}(TX)
      \arrow[l,"\gamma_{LX}",yshift=-1ex]
      \arrow[d,dashed,"\kant_{LX}"]
      \arrow[ddd,rounded corners, dashed, to path={--([xshift=2ex]\tikztostart.east) node[right,yshift=-12ex]{$\scriptstyle\be_{LX}$}
    |- (\tikztotarget.east)}]\\
    &
    |[alias=C]|
    \power (\alg{T}(\tilde FLX,(\tobject,o))) \arrow[r,"\alpha_{\tilde FLX}",yshift=1ex,"\perp"'] \arrow[dd,"\power((c^{\#})^\bullet)"] &
    |[alias=D]|
    \Psi^{\text{op}}(FTX)
    \arrow[l,"\gamma_{\tilde FLX}",yshift=-1ex] \arrow[dd,"(c^{\#})^*"]\\
    \power \Ccat(FTX,\tobject) \arrow[d,"\power(c^\bullet)"]\\
    \power \Ccat(X,\tobject) \arrow[r,"\alpha'_X",yshift=1ex,"\cong"']
    &\arrow[l,"\gamma'_X",yshift=-1ex] \power (\alg{T}(LX,(\tobject,o))) \arrow[r,"\alpha_{LX}", yshift=1ex,"\perp"']& \arrow[l,"\gamma_{LX}",yshift=-1ex] \Psi^{\text{op}}(TX)
  \end{tikzcd}
  \caption{The adjoint setup for algebras, where
    $\scriptstyle\be_{LX}=(c^\#)^*\circ \kant_{LX} \land
    \alpha_{LX}(\Theta_{LX})$ and
    $\scriptstyle\lo'_X=\power(c^\bullet) \circ \Lambda_X' \cup
    \Theta_X$.}\label{fig:AdjointSetupAlgebras}
\end{figure}

We can now invoke the results of the previous chapter and assume that
$\Lambda$ is compatible with the closure induced by the adjunction,
that is, we work without propositional operators (hence $\cl'$, as
mentioned in Assumption~\ref{assum:PropositionalClosure}, is the
identity), only constants, at first sight a strong property. We will
however see in the next section that this always holds when $F$ is a
machine functor and we choose suitable evaluation maps.

The next theorem focusses on free algebras and is partly a corollary
of \cref{prop:compatibility} and \cref{thm:genHML}. However there is a
new component: instead of defining the logic function on (free)
Eilenberg-Moore categories, reindexing via the determinized coalgebra
$c^\#$, it is possible -- as indicated above -- to define it directly
on arrows of type $X\to\tobject$ living in $\Ccat$, reindexing with
$c$. This coincides with the view that formulas should be evaluated on
states in $X$ rather than elements of $TX$. The diagram in
\cref{fig:AdjointSetupAlgebras} outlines how to show this result.

\begin{theoremrep}
  \label{thm:AdjLogHM-EMcase}
  We fix a coalgebra $(c\colon X\to FTX)\in\Ccat$.  Assume that
  $\Lambda_{LX}$ is compatible with the closure $\cl_{LX}$ and fix
  $\Theta_{LX}\subseteq \alg{T}({LX},\tobject)$ to ensure that \ref{assum:ConstantsExist} holds.
  \begin{enumerate}
    \item Then the logic function
  $\lo_{LX} = \power((c^\#)^\bullet)\circ \Lambda_{LX} \cup
    \Theta_{LX} $ is $\cl_{LX}$-compatible.
    \item For the behaviour function
  $ \be_{LX} = (c^\#)^*\circ \kant_{LX} \land\
  \alpha_{LX}(\Theta_{LX})$  (where
  $\kant_{LX} = \alpha_{\tilde{F}{LX}}\circ
  \Lambda_{LX}\circ\gamma_{LX}$), we have $\alpha_{LX}(\mu \lo_{LX}) = \nu \be_{LX}$.
    \item Now define another logic function
  $\lo'_X = \power(c^\bullet)\circ\Lambda'_X\cup \Theta'_X$ with
  $\Theta_{LX} = \alpha'_X(\Theta'_X)$. It holds that
  $ \alpha'_X\circ \lo'_X \circ \gamma'_X = \lo_{LX}$ and we obtain
  $\alpha_{LX}(\alpha_X' (\mu \lo'_X)) = \nu \be_X$.
  \end{enumerate}
\end{theoremrep}

\begin{proof}
  The functions $\lo_{LX}$, $\be_{LX}$ are instantiations of the
  templates in \cref{sec:adjointlogic}, where $\cl'_{LX}$ is
  the identity. Hence $\cl_{LX}$-compatibility of $\lo_{LX}$ follows
  directly from \cref{prop:compatibility} and preservation of
  fixpoints from \cref{thm:genHML}.

  In order to prove that
  $\alpha'_X\circ \lo'_X \circ \gamma'_X = \lo_{LX}$, we first show that
  $\alpha_X' \circ \power (c^\bullet) \circ \Lambda_X' \circ \gamma_X'
  = \power((c^\#)^\bullet) \circ \Lambda_{LX}$.  Since
  $\alpha'_X,\gamma'_X$ are inverse to each other (the adjunction is
  an equivalence), this is equivalent to showing
  $\power (c^\bullet) \circ \Lambda_X' = \gamma_X' \circ
  \power((c^\#)^\bullet) \circ \Lambda_{LX}\circ \alpha_X'$. Hence it
  is sufficient to show that
  $\gamma_X' \circ \power((c^\#)^\bullet) = \power(c^\bullet)\circ
  \power(R)$, since this implies
  \[
    \gamma_X' \circ \power((c^\#)^\bullet) \circ \Lambda_{LX}\circ
    \alpha_X' = \power(c^\bullet)\circ \power(R)\circ
    \Lambda_{LX}\circ \alpha_X' = \power (c^\bullet) \circ \Lambda_X'.
  \]
  Hence let $S\subseteq \alg{T}(\tilde{F}LX,(\tobject,o))$. The map
  $\power(c^\bullet)\circ \power(R)$ transforms all $g\in S$ (which are of
  type $g\colon FTX\to \tobject$ as arrows in $\Ccat$) to
  $g\circ c\colon X\to \tobject$ in $\Ccat$.  Instead
  $\gamma_X' \circ \power((c^\#)^\bullet)$ transforms such $g$ first
  to EM maps $g\circ c^\#\colon TX\to \tobject$ and then to
  $g\circ c^\#\circ \eta_X\colon X\to \tobject$ living in $\Ccat$ (note
  that $\eta_X$ is the unit of free-forgetful adjunction
  $L\dashv R\colon \Ccat \to \alg{T}$). Then
  $g\circ c = g\circ c^\#\circ \eta_X$ follows, since $c^\#$ is the
  transpose of $c$.

  From this $\alpha'_X\circ \lo'_X \circ \gamma'_X = \lo_X$ follows,
  since $\alpha'_X$ preserves unions and
  $\Theta_{LX} = \alpha'_X(\Theta'_X)$.

  This implies $\alpha'_X\circ \lo'_X = \lo_X\circ \alpha'_X$, since
  $\alpha'_X,\gamma'_X$ are inverse to each other. Hence
  $\alpha'_X(\mu \lo'_X) = \mu \lo_X$ according to
  \cref{thm:fixpoint-preservation}.

  We can conclude by observing that $\alpha_{LX}(\alpha'_X(\mu
  \lo'_X)) = \alpha_{LX}(\mu \lo_X) = \nu \be_X$.
\end{proof}

We hence consider a simple logic $\EMlog$ for $\alg{T}$, where $T$ is
a monad on $\set$:
\[
  \varphi\in \EMlog ::= \theta \ \mid\ [\lambda] \varphi \qquad
  (\text{where $\theta\in\Theta,\lambda\in\Lambda$})
\]
Given a coalgebra $c\colon X \to FTX \in \set$, each formula
$\varphi \in \EMlog$ is interpreted as a function
$\sem \varphi \colon X \to \tobject$, which is defined by
structural induction as follows:
\begin{itemize}
\item Let $\varphi = \theta$. Then $\sem \varphi$ is given by a
  predefined constant $X\to \tobject$. %
\item Let $\varphi = [\lambda]\varphi'$. Then
  $\sem \varphi = \ev_\lambda\circ Fo\circ FT\sem \varphi\circ c$
  (see definition of $\Lambda'_X$ in \cref{lem:lambda-prime}).
\end{itemize}

\begin{corollary}
  \label{cor:logic-expressive}
  Under the requirements of \cref{thm:AdjLogHM-EMcase}, the
  logic $\EMlog$ is expressive for the behavioural conformance
  $\be_{LX}$, i.e.,
  $
    \alpha_{LX}(\alpha_{X}'(\{\sem \varphi \mid \varphi\in\EMlog\})) =
    \nu \be_{LX}.
  $
\end{corollary}

\subsection{The Machine Functor}
\label{subsec:machine-functor}

Our next aim is to show that the machine functor has certain natural
evaluation maps ensuring that the predicate lifting is
$\cl$-compatible (one of the conditions of
\cref{thm:AdjLogHM-EMcase}). Throughout this section, we restrict
ourselves to a monad $T$ on $\set$ and fix the \emph{machine functor}
$M=\_^\Sigma \times B$ with $\Sigma\in\set$ and
$(B,b)\in\alg{T}$. Since all monads in $\set$ are strong and $B$ is
endowed with a $T$-algebra structure $b\colon TB\to B$, there is a
canonical distributive law $\zeta$
\cite[Exercise~5.4.4]{jacobs_coalgBook}:
\begin{equation}\label{eq:MachDLaw}
  \zeta_X\colon T(X^\Sigma \times B) \xrightarrow{\langle a \mapsto
    T(\pi_a\circ \pi_1),b\circ T\pi_2\rangle} (TX)^\Sigma \times B,
\end{equation}
where
$(\pi_i)_{i\in \{1,2\}}$ are the usual projections and
$\pi_a \colon X^\Sigma \to X$ is the evaluation map ($\pi_a(g) = g(a)$
where $g\colon \Sigma\to X$). Now let $\tilde M$ be
the lifting of $M$ to $\alg{T}$, induced by the
$\zeta$.
We observe that the evaluation maps suggested by it arise from natural
transformations in the sense of
\cref{lem:LambdaCompatibleForNaturalEvs}.

\begin{propositionrep}
  \label{prop:evalmaps}
  \begin{enumerate}
  \item\label{prop:evalmaps-actions} Let $a\in \Sigma$. Then
    $\eta_a\colon \tilde M\Rightarrow\text{Id}$ given by the
    composite
    $MX \xrightarrow{\pi_1} X^\Sigma \xrightarrow{\pi_a} X$ is a
    natural transformation.
  \item\label{prop:evalmaps-observations} Let
    $f\colon (B,b) \to (\tobject,o)$ be a homomorphism.  Then
    $\eta'_f\colon \tilde M \Rightarrow\mathbf{\tobject}$ given by the
    composite $MX \xrightarrow{\pi_2} B \xrightarrow f \tobject$ is
    a natural transformation.
  \end{enumerate}
  \noindent Thus, $\ev_a = \eta_\tobject$, $\ev_f = \eta'_\tobject$
  satisfy the properties of \cref{lem:LambdaCompatibleForNaturalEvs},
  and if each evaluation map is of this form, then $\Lambda$ is
  $\cl$-compatible.
\end{propositionrep}

\begin{proof}
  First remember that every monad on $\set$ is strong.

  Note that in both the cases it is straightforward to see that the
  evaluation maps are natural transformations of type
  $M \Rightarrow \text{Id}$ or $M \Rightarrow \mathbf{\tobject}$;
  thus, it suffices to show that $\ev_a$ and $\ev_f$ are algebra
  maps. This is because the lifted functor $\tilde M$ and $M$ are the
  same on arrows; thus, naturality of $\ev_a,\ev_f$ follows
  directly from the naturality squares in $\set$.

  For the first item, consider the following diagram with $(A,a)\in\alg{T}$:
  \[
  \begin{tikzcd}
    T(A^\Sigma \times B) \arrow[r,"T\pi_1"] \arrow[d,"\langle a \mapsto
    T(\pi_a \circ \pi_1){,}b\circ T\pi_2\rangle"'] & T(A^\Sigma)
    \arrow[r,"T\pi_a"] \arrow[d,"a \mapsto
    T\pi_a"'] %
    & TA \arrow[dd,"a"] \\
    (TA)^\Sigma \times B \arrow[d,"a^\Sigma \times B "'] & (TA)^\Sigma \arrow[d,"a^\Sigma"'] \arrow[ru,"\pi_{a}'"']\\
    A^\Sigma \times B \arrow[r,"\pi_1"] & A^\Sigma
    \arrow[r,"\pi_a"] & A
  \end{tikzcd}
  \]
  where $\pi_a,\pi_a'$ are the obvious evaluation maps. Note that the
  map $a \mapsto T\pi_a$ is the unique map that makes the triangle
  commute. Moreover $a^\Sigma$ is the unique map that makes the
  `lower' square commute. Lastly, the rectangle on the left
  commutes since
  $(a\mapsto T\pi_a ) \circ T\pi_1 = a\mapsto T(\pi_a\circ\pi_1)$.

  For the second item, consider the following diagram with
  $(A,a)\in\alg{T}$ and $(B,b)\xrightarrow
  f(\tobject,o)\in\alg{T}$. The square on the right commutes since $f$
  is an homomorphism, while the rectangle on the left obviously commutes.
  \[
  \begin{tikzcd}
    T(A^\Sigma \times B) \arrow[r,"T\pi_2"] \arrow[d,"\langle a\mapsto
    T(\pi_a \circ \pi_1){,}b\circ T\pi_2\rangle"'] & TB \arrow[r,"Tf"]
    \arrow[dd,"b"'] %
    & T\tobject \arrow[dd,"o"] \\
    (TA)^\Sigma \times B \arrow[d,"a^\Sigma \times \text{id}_B "'] & \\
    A^\Sigma \times B \arrow[r,"\pi_2"] & B \arrow[r,"f"] & \tobject
  \end{tikzcd}
  \]
\end{proof}

\subsection{Alternative Formulation of Kantorovich Lifting}
\label{subsec:alternative-kantorovich}

The behaviour function in \cref{subsec:logic-behaviour-algT} is based
on the generalized Kantorovich lifting~$\kant$~\cite{bbkk:coalgebraic-behavioral-metrics}, which
works as follows: given a pseudometric $d$ on $Y$ (here $Y=TX$),
generate \emph{all} non-expansive functions $Y\to \tobject$ wrt.\ $d$,
lift these functions to $FY\to\tobject$ and from there generate a
pseudometric on $FY$. However, $\kant_{LX}$ -- since it is defined in
an Eilenberg-Moore category -- works subtly differently: it takes
\emph{all non-expansive functions that are algebra
  homomorphisms}. This looks natural in the categorical setting, but
may pose problems if we implement the procedures. The standard
(probabilistic) Kantorovich lifting can for instance be computed based
on the Kantorovich-Rubinstein duality, by determining optimal
transport plans \cite{v:optimal-transport}.

Here, both types of liftings coincide at least on relevant
metrics. To show this result, we first define an alternative way of
lifting, as opposed to defining the lifting on $T$-algebra maps.
Applying \cref{thm:fibredadjoint} on $\Psi$ (rather
than on $\Psi\circ R$) gives the adjunction
$ \alpha_X^\Ccat \dashv \gamma_X^\Ccat \colon \power\Ccat(X,\tobject)
\rightleftarrows \Psi X^{\text{op}}.$ Now consider the lifting
$\kant^\Ccat_{TX} = \alpha_{FTX}^{\Ccat} \circ\Lambda_{TX} \circ
\gamma_{TX}^{\Ccat}$, where
$\Lambda_{TX}\colon\power \Ccat(TX,\tobject) \to \power
\Ccat(FTX,\tobject)$ is defined identically
to $\Lambda_{LX}$.

\begin{theoremrep}
  \label{thm:EquivalentKLiftings}
  Assume that $d$ is preserved by the co-closure, i.e.\
  $d=\alpha_{LX} (\gamma_{LX} (d))$, the co-closure
  $\alpha_X^\Ccat \circ \gamma_X^\Ccat$ is the identity, and each
  evaluation map $\ev_\lambda$ arises from some natural transformation
  either of type $F \Rightarrow \text{Id}$ or
  $F \Rightarrow \mathbf{\tobject}$. Then the two liftings
  $\kant_{LX},\kant^\Ccat_{TX}$ coincide on $d$, i.e.,
  $\kant_{LX}(d) = \kant^\Ccat_{TX}(d)$.
\end{theoremrep}

\begin{proof}
  Note that each evaluation map $\ev_\lambda$, $\lambda\in\Lambda$
  arises from a natural transformation either of type
  $F \Rightarrow \text{Id}$ or $F \Rightarrow \mathbf{\tobject}$. Let
  $K\subseteq \Lambda$ denote those predicate liftings that are of the
  second type.  Now consider the following derivation:
  \begin{align*}
    \kant_{LX} (d) &=\ \alpha_{\tilde{F}LX}(\Lambda_{LX}(\gamma_{LX}(d)))\\
    =&\ \alpha_{\tilde FLX} \left( \bigcup_{\lambda \in\Lambda} \power(\lambda_{LX})(\gamma_{LX}(d))\right)\\
    =&\ \alpha_{\tilde FLX}(\left\{\ev_{\lambda,\tobject} \circ \tilde Fh \mid h\in\gamma_{LX}(d) \land \lambda\in \Lambda \right\}) && \because \text{naturality of $\ev_{\lambda}$}\\
    =&\ \alpha_{\tilde FLX} \Big(\left\{ h \circ \ev_{\lambda,TX}
      \mid h\in\gamma_{LX}(d) \land \lambda \not\in K \right\} \\
    & \qquad\qquad \cup \left\{\ev_{\lambda,TX} \mid \lambda\in K \right\} \Big)&& \text{ definition of $\_^\bullet$}\\
    =&\ \alpha_{\tilde FLX} \Big(\left\{ \ev_{\lambda,TX}^\bullet(h)
      \mid h\in\gamma_{LX}(d) \land \lambda\not\in K \right\} \\
    & \qquad\qquad \cup \left\{\ev_{\lambda,TX} \mid \lambda\in K \right\} \Big)\\
    =&\ \alpha_{\tilde FLX}\Big( \bigcup_{\lambda\not\in K} \power(\ev_{\lambda,TX}^\bullet) (\gamma_{LX}(d)) \cup \bigcup_{\lambda\in K}\{\ev_{\lambda,TX}\} \Big) && \because \text{$\alpha$ is left adjoint}\\
    =&\ \bigwedge_{\lambda\not\in K} \alpha_{\tilde FLX}(\power(\ev_{\lambda,TX}^\bullet)(\gamma_{LX}(d))) \land \bigwedge_{\lambda\in K} \ev_{\lambda,TX}^*(d_\tobject) && \because \text{naturality of $\alpha$}\\
    =&\ \bigwedge_{\lambda\not\in K} \ev_{\lambda,TX}^* (\alpha_{LX}
    (\gamma_{LX}(d))) \land \bigwedge_{\lambda\in K}
    \ev_{\lambda,TX}^*(d_\tobject) && \because \text{$\alpha_{LX}(\gamma_{LX}(d)) = d$} \\
    =&\ \bigwedge_{\lambda\not\in K} \ev_{\lambda,TX}^* (d) \land \bigwedge_{\lambda\in K} \ev_{\lambda,TX}^*(d_\tobject) && \because \alpha_{TX}^\Ccat \circ \gamma_{TX}^\Ccat=\text{id}\\
    =&\ \bigwedge_{\lambda\not\in K} \ev_{\lambda,TX}^* (\alpha_{TX}^\Ccat (\gamma_{TX}^\Ccat(d))) \land \alpha_{FTX}^{\Ccat} \left(\{ \ev_{\lambda,TX}\mid\lambda\in K\}\right) &&\because \text{naturality of $\alpha^\Ccat$}\\
    =&\ \bigwedge_{\lambda\not\in K} \alpha^\Ccat_{FTX}(\power(\ev_{\lambda,TX}^\bullet)(\gamma_X^\Ccat(d))) \land \alpha_{FTX}^{\Ccat} \left(\{ \ev_{\lambda,TX}\mid\lambda\in K\}\right)\\
    =&\ \alpha^{\Ccat}_{FTX} \Big(\bigcup_{\lambda\not\in K} \power(\ev_{\lambda,TX}^\bullet)(\gamma_X^\Ccat(d)) \cup \{ \ev_{\lambda,TX}\mid\lambda\in K\} \Big)\\
    =&\ \alpha^{\Ccat}_{FTX} \circ \Lambda_{TX} \circ
    \gamma_X^{\Ccat}(d) \\
    =&\ \kant_{TX}^\Ccat(d). && \qedhere
  \end{align*}
\end{proof}

\begin{toappendix}
  \begin{lemmarep}
    \label{lem:co-closure-identity}
    Assume that $\Ccat = \set$ and $\Psi X = \dpmet_\V(X)$ (resp.\
    $\Psi = \pmet_\V(X)$) for an integral quantale $\tobject$. Let let
    $d_\tobject = [\_,\_]$ (resp.\ the symmetrized variant of
    $[\_,\_]$). Then $\alpha_X^\Ccat\circ\gamma_X^\Ccat$ is the
    identity.

  \end{lemmarep}

  \begin{proof}
    In order to show that
    $\alpha_X^\Ccat\circ\gamma_X^\Ccat = \mathit{id}_{\Psi X}$, we
    first observe that
    $\alpha_X^\Ccat\circ\gamma_X^\Ccat \preceq \mathit{id}_{\Psi
      (TX)}$ follows from the properties of the adjunction (Galois
    connection). In order to show equality, i.e.,
    $\alpha_X^\Ccat(\gamma_X^\Ccat(d)) = d$ for $d\in\Psi(TX)$, we
    prove that -- given $t_1,t_2\in TX$ -- there exists
    $h\in \gamma_X^\Ccat(d)$ such that
    $d(t_1,t_2) = h^*d_\tobject(t_1,t_2)$. According to
    \cref{lem:non-expansive-function-from-d} $h\colon TX\to \tobject$,
    defined as $h(t) = d(t_1,t)$, is non-expansive wrt.\ $d$ and hence
    contained in $\gamma_X^\Ccat(d)$. Furthermore
    $h^*d_\tobject(t_1,t_2) = d_\tobject(h(t_1),h(t_2)) = d(t_1,t_2)$,
    using \cref{lem:non-expansive-function-from-d}.
  \end{proof}
\end{toappendix}

The condition $d=\alpha_{LX} (\gamma_{LX} (d))$ is a necessary, but
not a serious restriction: this property is typically satisfied by the
$\top$ metric and is preserved by the behaviour function.
Hence during fixpoint iteration this invariant is preserved and the
greatest fixpoints of $\be_X$ based on either version of the
Kantorovich lifting coincide (if the fixpoint is reached in $\omega$
steps).

In addition, if $\Ccat = \set$
and %
$\V$ is an integral quantale, it can easily be shown that the
co-closure $\alpha_X^\Ccat \circ \gamma_X^\Ccat$ is the identity (see
\cref{lem:co-closure-identity} in the appendix).
This enables us to concretely spell out the behaviour function for the
case of the machine functor $M = \_^\Sigma\times B$, provided that the
conformances $\Psi X$ are (directed) pseudometrics.

\begin{theoremrep}
  \label{thm:beh-machine-functor}
  Assume that $\Ccat = \set$, $\Psi X = \dpmet_\V(X)$ (resp.\
  $\Psi X = \pmet_\V(X)$) for an integral quantale $\V$ and
  let $d_\V = [\_,\_]$ (resp.\ the symmetrized variant of
  $[\_,\_]$). Let $d\colon LX\times LX\to \V$ be a pseudometric
  that is preserved by the co-closure
  $\alpha_{LX}\circ\gamma_{LX}$.
  Assume that $M$ is the machine functor and the family of
  evaluation maps is
  \[ \{\ev_a\mid a\in \Sigma\}\cup \{\ev_f\mid f\in
    \mathcal{F}\subseteq \alg{T}((B,b),(\V,o)) \}. \]

  Then the corresponding behaviour function
  $\be_{LX}\colon \Psi (LX) \to \Psi (LX)$ is defined as follows: let
  $t_1,t_2\in LX$ with $c^\#(t_i) = (b_i,g_i)\in B\times LX^\Sigma$:
  \[ \be_{LX}(d)(t_1,t_2) = \bigwedge_{a\in \Sigma}
    d(g_1(a),g_2(a)) \land \bigwedge_{f\in \mathcal{F}}
    d_\V(f(b_1),f(b_2)) \land \bigwedge_{\theta\in\Theta_{LX}}
    d_\V(\theta(t_1),\theta(t_2)), \]
\end{theoremrep}

\begin{proof}
  According to our previous results (\cref{thm:genHML}) the behaviour
  function has the form
  $\be_{LX} = (c^\#)^*\circ\kant_{LX} \land \ \alpha_{LX}(\Theta_{LX})$.

  With the definition of $\alpha$ in \cref{thm:fibredadjoint} is
  straightforward to see that
  \[ \alpha_{LX}(\Theta_{LX})(t_1,t_2) = \bigwedge_{\theta\in\Theta_{LX}}
    \theta^*(d_\V)(t_1,t_2) = \bigwedge_{\theta\in\Theta_{LX}}
    d_\V(\theta(t_1),\theta(t_2)), \]
  matching the third part of the definition of $\be_{LX}$ given in the
  statement of the proposition.

  For the first and second part we have to unravel $\kant_{LX}$:
  according to \cref{thm:EquivalentKLiftings} we have
  \[ \kant_{LX} = \kant^\Ccat_{TX} = \alpha_{FTX}^\Ccat \circ
    \Lambda_{TX} \circ \gamma_{TX}^\Ccat. \] Hence, given
  $s_1,s_2\in TX$ where $s_i = (g_i,b_i)$, $g\colon \Sigma\to TX$,
  $b_i\in B$:
  \begin{align*}
    \kant_{LX}(d)(s_1,s_2) &= \alpha_{FTX}^\Ccat(\Lambda_{TX}(\gamma_{TX}^\Ccat(d)))(s_1,s_2) \\
    &= \bigwedge \{ k^*(d_\V)(s_1,s_2) \mid k\in
    \Lambda_{TX}(\gamma_{TX}^\Ccat(d)) \} \\
    &= \bigwedge \{ d_\V(k(s_1),k(s_2)) \mid k = \ev_\lambda\circ
    Fh, h\in \gamma_{TX}^\Ccat(d) \} \\
    &= \bigwedge \{ d_\V(\ev_\lambda\circ Fh(s_1),\ev_\lambda\circ
    Fh(s_2)) \mid h\colon TX\to \V, d\preceq h^* d_\V \} \\
    &= \bigwedge_{a\in\Sigma} \{ d_\V(\ev_a\circ
    Fh(s_1),\ev_a\circ Fh(s_2)) \mid h\colon TX\to \V, d\preceq h^* d_\V \} \land
    \mbox{} \\
    & \qquad \bigwedge_{f\in\mathcal{F}} \{ d_\V(\ev_f\circ
    Fh(s_1),\ev_f\circ
    Fh(s_2)) \mid h\colon TX\to \V, d\preceq h^* d_\V \}
  \end{align*}
  The property $d\preceq h^* d_\V$ states that $h$ is non-expansive
  (wrt.\ $d,d_\V$). We consider the two parts of the meet
  separately:
  \begin{itemize}
  \item First part of meet: assume that $\lambda=a\in
    \Sigma$. Then
    \[ \ev_a\circ Fh(s_i) = \ev_a\circ Fh(g_i,b_i) =
      h(g_i(a)) \] By non-expansiveness of $h$ it is clear that
    always
    $d_\V(h(g_1(a)),h(g_2(a))) \ge
    d(g_1(a),g_2(a))$. In fact we can choose $h$ such that
    equality holds: we set $h(x) = d(g_1(a),\_)$ and infer from
    \cref{lem:non-expansive-function-from-d} that it is
    non-expansive, i.e., that $d\preceq h^* d_\V$. Observe that
    $h(g_1(a)) = d(g_1(a),g_1(a)) = 1$ (remember that
    $d$ is reflexive and the quantale is integral). In addition
    $h(g_2(a)) = d(g_1(a),g_2(a))$. Hence we obtain
    \[ d_\V(h(g_1(a)),h(g_2(a))) =
      d(g_1(a),g_2(a)), \] again with
    \cref{lem:non-expansive-function-from-d}.

    Hence the infimum is reached in $d(g_1(a),g_2(a))$ and
    the first part of the meet simplifies to
    \[ \bigwedge_{a\in\Sigma} d(g_1(a),g_2(\sigma)). \]
  \item Second part of meet: assume that $\lambda=f\in
    \mathcal{F}$. Then
    \[ \ev_f\circ Fh(s_i) = \ev_f\circ Fh(g_i,b_i) = f(b_i). \]
    This means that the infimum is independent of the choice of $h$
    and simplifies to
    \[ \bigwedge_{f\in\mathcal{F}} d_\V(f(b_1),f(b_2)). \]
  \end{itemize}
  Now set $s_i = c^\#(t_i)$ to include the reindexing via $c^\#$
  and one obtains the characterization of $\be_{LX}$.
\end{proof}
\noindent
The above function $\be$ %
is co-continuous and fixpoint
iteration terminates after $\omega$ steps.

\section{Case Studies for the Linear-time Case}
\label{sec:examples}

\subsection{Workflow}

We recall the parameters of our framework and set out a workflow
that we follow in our case studies.
Let $F$ be a machine functor and $T$ a monad on a category $\Ccat$.
\begin{itemize}
\item Model systems as coalgebras of type
  $c\colon X\to FTX$ with a
  distributive law $\zeta\colon TF\Rightarrow FT$.
\item Fix a truth value object $(\tobject,o)\in \alg{T}$ and  $d_\tobject\in \Psi\tobject$.
\item Define a fibration (indexed category)
  $\Phi = \Psi\circ R\colon \alg{T}^{\text{op}} \to
  \pos$ by fixing an indexed category
  $\Psi\colon \Ccat^{\text{op}}\to \pos$ to define the
  conformances.
\item Fix a set $\Lambda$ of evaluation maps (predicate liftings) as
  homomorphisms $\tilde{F}(\tobject,o)\to (\tobject,o)$.
\item Fix a set of constants
  $\Theta_X\subseteq \Ccat(X,\tobject)$. %
\end{itemize}

Note that the last four conditions correspond to
Assumptions~\ref{assum:ConfExists}-\ref{assum:PropositionalClosure}, %
which are necessary to set up a logic and a behaviour function $\be$
as defined in \cref{subsec:logic-behaviour-algT} and guarantee
expressiveness of the resulting logic. Whenever we choose
$\Psi X = \dpmet_\V(X)$ or $\Psi X = \pmet_\V(X)$ for an integral
quantale $\V$, we can rely on the characterization of the fixpoint
equation in \cref{thm:beh-machine-functor}.

\newcommand{\CasestudyOne}{
\subsection{Trace Equivalence for Labelled Transition Systems}
\label{subsec:trace-equivalence}

We spell out a simple case study: trace equivalence
\cite{g:linear-branching-time} for labelled transition systems. The
main ingredients are summarized in the table below.

\medskip

\noindent
\begin{tabular}{|l|l|}
  \hline
  $\Ccat = \set$, $F=\_^\Sigma$, $T=\power$ & \textbf{Logic:} \\
  \qquad ($B=1$) & evaluation maps: $\ev_a$
  (Prop.~\ref{prop:evalmaps}(\ref{prop:evalmaps-actions})) \\
  $c\colon X\to (\power X)^\Sigma$ & constants $\Theta_X = \{1\}$,
  constant $1$-function \\
  $c^\#\colon \power X\to (\power X)^\Sigma$ & formulas:
  $\varphi = [a_1]\cdots [a_n]1$ \\ \hline $\tobject=\mathbf{2}$
  (Ex.~\ref{ex:quantale}(\ref{ex:quantale-boolean}))
  & \textbf{Behaviour function:} \\
  $o\colon\power\mathbf{2}\to \mathbf{2}$ supremum &
  $\be_{\power X}\colon \Psi(\power X) \to \Psi (\power X)
  $ \\
  $\Psi(X)$: equivalences on $X$ &
  $\be_{\power X}(R)(U,V) = \left(U=\emptyset \Leftrightarrow
    V=\emptyset\right) \land$ \\
  $d_\tobject$: equality on $\tobject$ &
  \qquad\qquad\qquad\qquad\qquad
  $\forall_{a\in\Sigma} (c^\#(U)(a),c^\#(V)(a))\in R$ \\\hline
\end{tabular}

\medskip

The modality $[a]$ boils down to the standard diamond modality (due to
the definition of $o$); a state $x\in X$ satisfies
$\varphi = [a_1]\cdots [a_n]1$ iff there exists a trace
$a_1\cdots a_n$ from $x$. The constant~$1$ is needed to start building
formulas and to distinguish the empty set from a non-empty set on
$LX = \power X$. (Note that
$\Theta_{LX} = \alpha'_X(\Theta_X) = \{\tilde{1}\}$ with
$\tilde{1}(Y)=0$ iff $Y=\emptyset$.)
Its role cannot be taken by a constant modality or an operator, since
those have to be homomorphisms in $\alg{\power}$, hence
sup-preserving.

Expressiveness of trace logic $\EMlog$ now directly follows from
\cref{cor:logic-expressive}.

}

\newcommand{\CasestudyTwo}{
\subsection{Trace Distance for Probabilistic Automata}
\label{subsec:trace-prob-automata}

A \emph{probabilistic automaton} \cite{RABIN1963230} is a quadruple
$(X, \Sigma, \mu, p)$ where for each state $x\in X$ and each possible
action $a\in \Sigma$ there is a probability distribution $\mu_{x,a}$
on the possible successors in $X$, and where each state $x\in X$ has a
payoff value $p(x) \in [0,1]$. Following
\cite{sbbr:powerset-coalgebraically,sbbr:generalizing-determinization-journal},
we model them as coalgebras in the Eilenberg-Moore setting as detailed
in the table below.

\medskip

\noindent
\begin{tabular}{|l|l|}
  \hline
  $\Ccat = \set$, $F=\_^\Sigma\times [0,1]$, $T=\dist$ & \textbf{Logic:} \\
  \qquad ($B=[0,1]$, $b$ expectation) & evaluation maps: $\ev_a$
  (Prop.~\ref{prop:evalmaps}(\ref{prop:evalmaps-actions})) \\
  $c\colon X\to (\dist X)^\Sigma\times[0,1]$ &
  \qquad\qquad\qquad\qquad $\ev_*(f,r)= r$
  (Prop.~\ref{prop:evalmaps}(\ref{prop:evalmaps-observations})) \\
  $c^\#\colon \dist X\to (\dist X)^\Sigma\times [0,1]$ & constants
  $\Theta_X = \emptyset$ \\
  & formulas: $\varphi = [a_1]\cdots [a_n]*$ \\ \hline
  $\tobject=[0,1]$ (Ex.~\ref{ex:quantale}(\ref{ex:quantale-reals})) &
  \textbf{Behaviour function:} \\
  $o\colon \dist[0,1]\to [0,1]$ expectation &
  $\be_{\dist X}\colon \Psi(\dist X)  \to \Psi(\dist X)$ \\
  $\Psi(X) = \pmet_{[0,1]}(X)$ &
  $\be_{\dist X}(d)(p_1,p_2) = \max\{\sup_{a\in \Sigma}
  d(g_1(a),g_2(a)), d_\tobject(r_1,r_2) \}$ \\
  $d_\tobject(r,s) = |r-s|$ & \\ \hline
\end{tabular}

\medskip

Thus, given a formula $\varphi = [a_1]\dots [a_n]*$ and a state
$x\in X$, $\llbracket \varphi\rrbracket(x)$ gives us the expected
payoff after choosing actions according to the word $a_1\dots
a_n$. The distance of two states $x_1,x_2$ is hence the supremum of
the difference of payoffs, over all words. %

Expressiveness again follows from \cref{cor:logic-expressive}.
}

\newcommand{\CasestudyThree}{
\subsection{Directed Fuzzy Trace Distance} %
\label{subsec:fuzzy-mts}

We now consider directed trace distances for weighted transition
systems over a generic quantale.

We work with the ``fuzzy'' monad $T = \mathcal{P}_\V$ (aka $\V$-valued
powerset monad \cite[Remark~1.2.3]{hofmann_seal_tholen_2014}) on
$\set$ that is defined as $\mathcal{P}_\V = \V^X$ on objects and as
$Tf(g)(y) = \bigvee_{f(x)=y} g(x)$ (for $f\colon X\to Y$) on
arrows.  Its unit
$\eta_X\colon X\to\mathcal{P}_\V X$ is given by $\eta_X(x)(x') = 1$ if $x=x'$ and $0$ (the empty join) otherwise. Multiplication
$\mu_X\colon \mathcal{P}_\V \mathcal{P}_\V X \to \mathcal{P}_\V X$ is
defined as
$\mu_X(G)(x)=\bigvee_{g\in \mathcal{P}_\V X} G(g)\otimes g(x)$. Note
that for $\V=2$ (cf.~\cref{ex:quantale}(\ref{ex:quantale-boolean}))
$T$ corresponds to the powerset monad $\power$.

\medskip

\noindent
\begin{tabular}{|l|l|}
  \hline
  $\Ccat = \set$, $F=\_^\Sigma$, $T=\mathcal{P}_\V$ & \textbf{Logic:} \\
  \qquad ($B=1$) & evaluation maps: $\ev_a$
  (Prop.~\ref{prop:evalmaps}(\ref{prop:evalmaps-actions})) \\
  $c\colon X\to (\mathcal{P}_\V X)^\Sigma$ & constants
  $\Theta_X = \{1\}$, constant $1$-function \\
  $c^\#\colon \mathcal{P}_\V X\to (\mathcal{P}_\V X)^\Sigma$ & formulas: $\varphi =
  [a_1]\cdots [a_n]1$ \\
  \hline $\tobject=\V$, $o\colon \mathcal{P}_\V \V \to \V$ &
  \textbf{Behaviour function:} \\
  \qquad\qquad $g\mapsto \bigvee_{v\in \V} g(v)\otimes v$ &
  $\be_{\mathcal{P}_\V X}\colon \Psi(\mathcal{P}_\V X )  \to \Psi(\mathcal{P}_\V X )$ \\
  $\Psi(X) = \dpmet_{\V}(X)$ &
  $\be_{\mathcal{P}_\V X}(d)(g_1,g_2) =  \bigwedge_{a\in \Sigma}
  d(c^\#(g_1)(a),c^\#(g_2)(a))\ \land$ \\
  $d_\tobject(v,v') = [v,v']$ & \qquad\qquad\qquad\qquad\qquad
  $\left[\bigvee_{x\in X} g_1(x),
  \bigvee_{x\in X} g_2(x)\right]$ \mystrut \\
  \hline
\end{tabular}

\medskip

Evaluating a formula $\varphi = [a_1]\dots [a_n]1$ on a state
$x_0\in X$ results in
\[ \sem\varphi(x_0) = \bigvee \{ \bigotimes_{0\leq i <n}
  c(x_i)(a_{i+1})(x_{i+1}) \mid x_1\dots x_n\in X^n \}. \]
This follows directly from structural induction on $\varphi$, from
distributivity and from the evaluation of the modality $[a]\varphi'$:
\[
\sem{[a]\varphi'}(x) = \bigvee_{y\in X} c(x)(a)(y) \otimes \sem{\varphi'}(y).
\]
Intuitively we check how well $x$ can match the trace $a_1\dots a_n$,
where $c(x)(a)(y)$ measures the degree to which $x$ can make an
$a$-transition to $y$.

The second part of the minimum in the definition of $\be$ stems from
the constants $\Theta_X = \{1\}$, since
$\Theta_{LX} = \alpha'_X(\Theta_X) = \{\tilde{1}\}$ with
$\tilde{1}(h) = \bigvee_{x\in X} h(x)$ for $h\colon X\to
\V$. Without it, the fixpoint iteration would stabilize at the
constant $1$-pseudometric.

Expressiveness again follows from
\cref{cor:logic-expressive}. Expressiveness of a logic for
\emph{symmetric} fuzzy trace distance has already been shown in
previous work~\cite{fswbgkm:graded-logics-em}.
}

\CasestudyOne
\CasestudyTwo
\CasestudyThree

\section{Conclusion}
\label{sec:conclusion}

\noindent\emph{Related work.} By now there is a large number of papers
considering coalgebraic semantics beyond branching-time, for instance
\cite{hjs:generic-trace-coinduction,sbbr:powerset-coalgebraically,mps:trace-semantics-graded-monads,c:branching-linear-coalgebra}.
Furthermore in the same period quite a wealth of results on the
treatment of behavioural metrics in coalgebraic generality has been
published
\cite{bw:behavioural-pseudometric,bbkk:coalgebraic-behavioral-metrics,kkkrh:expressivity-quantitative-modal-logics,Forster_et_al:CSL.2023:Density}. However,
there is little work combining both linear-time semantics and
behavioural metrics in the setting of coalgebra. In that respect we
want to mention \cite{fswbgkm:graded-logics-em} that is based on the
graded monad framework \cite{mps:trace-semantics-graded-monads} and
which investigates exactly this combination. However, different from
the present paper, the focus is on the expressiveness of the logics
with respect a graded semantics (that intuitively specifies the traces
of a state). Hence, using the classification of the introduction, it
studies the relationship of~(i) and~(ii).

Unlike other approaches, our main focus is on exploiting an adjunction
(Galois connection) and fixpoint preservation results to obtain
Hennessy-Milner theorems ``for free''. We start by setting up a logic,
characterizing the behavioural equivalence, and investigate under
which circumstances we can derive a corresponding fixpoint
characterization. The fixpoint equation might be defined on an infinite state
space, but often there are finitary techniques that can be employed,
such as reducing the state space to a finite subset, linear
programming, up-to techniques, etc. In particular, for systems as in
\cref{subsec:fuzzy-mts} we are working on promising results (based on
\cite{bkp:abstraction-up-to-games-fixpoint}) for deriving bounds for
behavioural distances via finite witnesses using up-to techniques,
even for infinite state spaces. The algorithmic angle of our approach
is not yet fully worked out in the present paper but establishing
fixpoint equations as we do here is a necessary first step in this
direction.

Note that our concept deviates %
from the dual adjunction approach
\cite{k:coalgebraic-logic-beyond-sets,kp:coalgebraic-modal-logics-overview,kr:logics-coinductive-predicates,Pavlovic2006:DualAdj}
to coalgebraic modal logic. There the functor on the ``logic
universe'' characterizes the \emph{syntax} of the logics, while the
semantics is instead given by a natural transformation. Nevertheless, it complements (at least when restricted to the classical case of Boolean predicates) the recent approach \cite{turkenburg_et_al:LIPIcs.CALCO.2023.6} that combines fibrations in the dual adjunction setup since having contravariant Galois connections between fibres (at a `local' level) is equivalent to having dual adjunctions between certain fibred categories (at a `global' level). It is unclear how to establish this correspondence in the setting of quantitative $\V$-valued predicates.

\smallskip

\noindent\emph{Future work.} Currently the operators of the logic,
given by $\cl'$, are rather generic, although we instantiated them in
special cases to ensure expressiveness (see in particular
Sections~\ref{sec:quantale-branching-logic}
and~\ref{subsec:machine-functor}). We envision a general theory to
ensure expressiveness of the logics, similar to Post's functional
completeness theorem \cite{pm:post-functional-completeness}, which
characterizes complete sets of operators for the boolean case. This
question is strongly related to the notion of an approximating family
in \cite{kkkrh:expressivity-quantitative-modal-logics} that has again
close connections to compatibility as discussed in
\cite{bgkm:hennessy-milner-galois}.

We will also study the condition requiring that a conformance
(pseudometric) is preserved by the co-closure
($d=\alpha_{LX} (\gamma_{LX} (d))$). Previous results
\cite{bgkm:hennessy-milner-galois} suggest that this is related to the
notion of (metric) congruence, as e.g.\ defined in
\cite{Bonchi_et_al:PowerConvexAlg}, but the connection seems to be
non-trivial.

Another avenue of research is to further investigate the
quantale-valued logic for the branching case introduced in
\cref{sec:quantale-branching-logic}, to extend it to the undirected
case and restrict to finitary operators.

\bibliographystyle{plain} \bibliography{references}

\begin{thebibliography}{10}

\bibitem{AHS90}
Ji{\v{r}\'i} Ad{\'a}mek, Horst Herrlich, and George~E. Strecker.
\newblock {\em Abstract and concrete categories: {T}he joy of cats}.
\newblock Wiley, 1990.
\newblock Republished in: Reprints in Theory and Applications of Categories,
  No. 17 (2006) pp.~1--507.

\bibitem{bbkk:coalgebraic-behavioral-metrics}
Paolo Baldan, Filippo Bonchi, Henning Kerstan, and Barbara K{\"o}nig.
\newblock Coalgebraic behavioral metrics.
\newblock {\em Logical Methods in Computer Science}, 14(3), 2018.
\newblock Selected Papers of the 6th Conference on Algebra and Coalgebra in
  Computer Science (CALCO 2015).

\bibitem{bkp:abstraction-up-to-games-fixpoint}
Paolo Baldan, Barbara K\"onig, and Tommaso Padoan.
\newblock Abstraction, up-to techniques and games for systems of fixpoint
  equations.
\newblock In {\em Proc. of CONCUR '20}, volume 171 of {\em {LIPIcs}}, pages
  25:1--25:20. Schloss Dagstuhl -- Leibniz Center for Informatics, 2020.

\bibitem{bgkm:hennessy-milner-galois}
Harsh Beohar, Sebastian Gurke, Barbara K\"onig, and Karla Messing.
\newblock {H}ennessy-{M}ilner theorems via {G}alois connections.
\newblock In {\em Proc. of CSL '23}, volume 252 of {\em {LIPIcs}}, pages
  12:1--12:18. Schloss Dagstuhl -- Leibniz Center for Informatics, 2023.

\bibitem{Bonchi_et_al:PowerConvexAlg}
Filippo Bonchi, Alexandra Silva, and Ana Sokolova.
\newblock {The Power of Convex Algebras}.
\newblock In {\em Proc. of CONCUR '17}, volume~85 of {\em LIPIcs}, pages
  23:1--23:18. Schloss Dagstuhl--Leibniz-Zentrum fuer Informatik, 2017.

\bibitem{c:branching-linear-coalgebra}
Corina C{\^{\i}}rstea.
\newblock From branching to linear time, coalgebraically.
\newblock {\em Fundamenta Informaticae}, 150(3-4):379--406, 2017.

\bibitem{cc:systematic-analysis}
Patrick Cousot and Radhia Cousot.
\newblock Systematic design of program analysis frameworks.
\newblock In {\em Proc. of POPL '79}, pages 269--282. ACM Press, 1979.

\bibitem{cc:temporal-abstract-interpretation}
Patrick Cousot and Radhia Cousot.
\newblock Temporal abstract interpretation.
\newblock In Mark~N. Wegman and Thomas~W. Reps, editors, {\em Proc. of POPL
  '00}, pages 12--25. {ACM}, 2000.

\bibitem{dgjp:metrics-labelled-markov}
Jos{\'e}e Desharnais, Vineet Gupta, Radha Jagadeesan, and Prakash Panangaden.
\newblock Metrics for labelled {M}arkov processes.
\newblock {\em Theoretical Computer Science}, 318:323--354, 2004.

\bibitem{DorschEA19}
Ulrich Dorsch, Stefan Milius, and Lutz Schr{\"{o}}der.
\newblock Graded monads and graded logics for the linear time - branching time
  spectrum.
\newblock In Wan~J. Fokkink and Rob van Glabbeek, editors, {\em Concurrency
  Theory, {CONCUR} 2019}, volume 140 of {\em LIPIcs}, pages 36:1--36:16.
  Schloss Dagstuhl -- Leibniz-Zentrum f{\"{u}}r Informatik, 2019.

\bibitem{Forster_et_al:CSL.2023:Density}
Jonas Forster, Sergey Goncharov, Dirk Hofmann, Pedro Nora, Lutz Schr\"{o}der,
  and Paul Wild.
\newblock {Quantitative Hennessy-Milner Theorems via Notions of Density}.
\newblock In {\em Proc. of CSL '23}, volume 252 of {\em {LIPIcs}}, pages
  22:1--22:20, 2023.

\bibitem{fswbgkm:graded-logics-em}
Jonas Forster, Lutz Schröder, Paul Wild, Harsh Beohar, Sebastian Gurke,
  Barbara König, and Karla Messing.
\newblock Graded semantics and graded logics for {Eilenberg-Moore} coalgebras,
  2023.
\newblock arXiv:2307.14826.

\bibitem{GiacaloneEA90}
Alessandro Giacalone, Chi{-}Chang Jou, and Scott~A. Smolka.
\newblock Algebraic reasoning for probabilistic concurrent systems.
\newblock In Manfred Broy and Cliff~B. Jones, editors, {\em Programming
  concepts and methods}, pages 443--458. North-Holland, 1990.

\bibitem{hjs:generic-trace-coinduction}
Ichiro Hasuo, Bart Jacobs, and Ana Sokolova.
\newblock Generic trace semantics via coinduction.
\newblock {\em Logical Methods in Computer Science}, 3(4:11):1--36, 2007.

\bibitem{hofmann_seal_tholen_2014}
Dirk Hofmann, Gavin~J. Seal, and Walter Tholen, editors.
\newblock {\em Lax algebras}, page 145–283.
\newblock Encyclopedia of Mathematics and its Applications. Cambridge
  University Press, 2014.

\bibitem{jacobs-fibrations}
Bart Jacobs.
\newblock {\em Categorical Logic and Type Theory}, volume 141 of {\em Studies
  in Logic and the Foundations of Mathematics}.
\newblock Elsevier, 1st edition, Jan 1999.

\bibitem{Jacobs_indexedcat}
Bart Jacobs.
\newblock Predicate logic for functors and monads.
\newblock Available from author's website, 2010.

\bibitem{jacobs_coalgBook}
Bart Jacobs.
\newblock {\em Introduction to Coalgebra: Towards Mathematics of States and
  Observation}.
\newblock Cambridge Tracts in Theoretical Computer Science. Cambridge
  University Press, 2016.

\bibitem{jss:trace-determinization}
Bart Jacobs, Alexandra Silva, and Ana Sokolova.
\newblock Trace semantics via determinization.
\newblock In {\em Proc. of CMCS '12}, pages 109--129. Springer, 2012.
\newblock {LNCS} 7399.

\bibitem{k:coalgebraic-logic-beyond-sets}
Bartek Klin.
\newblock Coalgebraic modal logic beyond sets.
\newblock In {\em Proc. of MFPS '07}, volume 173 of {\em ENTCS}, pages
  177--201, 2007.

\bibitem{CodensityGames}
Yuichi Komorida, Shin{-}ya Katsumata, Nick Hu, Bartek Klin, and Ichiro Hasuo.
\newblock Codensity games for bisimilarity.
\newblock In {\em Proc. of LICS '19}, pages 1--13, 2019.

\bibitem{kkkrh:expressivity-quantitative-modal-logics}
Yuichi Komorida, Shin{-}ya Katsumata, Clemens Kupke, Jurriaan Rot, and Ichiro
  Hasuo.
\newblock Expressivity of quantitative modal logics: Categorical foundations
  via codensity and approximation.
\newblock In {\em Proc. LICS '21}, pages 1--14. {IEEE}, 2021.

\bibitem{kp:coalgebraic-modal-logics-overview}
Clemens Kupke and Dirk Pattinson.
\newblock Coalgebraic semantics of modal logics: An overview.
\newblock {\em Theoretical Computer Science}, 412:5070--5094, 2011.

\bibitem{kr:logics-coinductive-predicates}
Clemens Kupke and Jurriaan Rot.
\newblock Expressive logics for coinductive predicates.
\newblock In {\em Proc. of CSL '20}, volume 152 of {\em LIPIcs}, pages
  26:1--26:18. Schloss Dagstuhl -- Leibniz Center for Informatics, 2020.

\bibitem{Law73}
F.~William Lawvere.
\newblock Metric spaces, generalized logic, and closed categories.
\newblock {\em Rendiconti del Seminario Matemàtico e Fisico di Milano},
  43(1):135--166, 1973.
\newblock Republished in Reprints in Theory and Applications of Categories 1
  (2002), 1--37.

\bibitem{mps:trace-semantics-graded-monads}
Stefan Milius, Dirk Pattinson, and Lutz Schr\"oder.
\newblock Generic trace semantics and graded monads.
\newblock In {\em Proc. of CALCO '15}, volume~35 of {\em LIPIcs}, pages
  253--269. Schloss Dagstuhl -- Leibniz-Zentrum fuer Informatik, 2015.

\bibitem{Pavlovic2006:DualAdj}
Dusko Pavlovic, Michael Mislove, and James Worrell.
\newblock Testing semantics: Connecting processes and process logics.
\newblock In {\em Proc. of AMAST '06}, pages 308--322. Springer, 2006.
\newblock {LNCS} 4019.

\bibitem{pm:post-functional-completeness}
Francis~Jeffry Pelletier and Norman~M. Martin.
\newblock Post's functional completeness theorem.
\newblock {\em Notre Dame Journal of Formal Logic}, 31(2), 1990.

\bibitem{p:complete-lattices-up-to}
Damien Pous.
\newblock Complete lattices and up-to techniques.
\newblock In {\em Proc. of APLAS '07}, pages 351--366. Springer, 2007.
\newblock {LNCS} 4807.

\bibitem{RABIN1963230}
Michael~O. Rabin.
\newblock Probabilistic automata.
\newblock {\em Information and Control}, 6(3):230--245, 1963.

\bibitem{Rutten00}
Jan Rutten.
\newblock Universal coalgebra: A theory of systems.
\newblock {\em Theoretical Computer Science}, 249:3--80, 2000.

\bibitem{Schroeder2007:ExpressivityCoalgModalLog}
Lutz Schr\"{o}der.
\newblock Expressivity of coalgebraic modal logic: The limits and beyond.
\newblock {\em Theoretical Computer Science}, 390(2):230--247, 2008.
\newblock Foundations of Software Science and Computational Structures.

\bibitem{sbbr:powerset-coalgebraically}
Alexandra Silva, Filippo Bonchi, Marcello~M. Bonsangue, and Jan J. M.~M.
  Rutten.
\newblock Generalizing the powerset construction, coalgebraically.
\newblock In {\em Proc. of FSTTCS '10}, volume~8, pages 272--283. Schloss
  Dagstuhl -- Leibniz Center for Informatics, 2010.
\newblock {LIPIcs}.

\bibitem{sbbr:generalizing-determinization-journal}
Alexandra Silva, Filippo Bonchi, Marcello~M. Bonsangue, and Jan J. M.~M.
  Rutten.
\newblock Generalizing determinization from automata to coalgebras.
\newblock {\em Logical Methods in Computer Science}, 9(1:09), 2013.

\bibitem{t:lattice-fixed-point}
Alfred Tarski.
\newblock A lattice-theoretical fixpoint theorem and its applications.
\newblock {\em Pacific Journal of Mathematics}, 5:285--309, 1955.

\bibitem{turkenburg_et_al:LIPIcs.CALCO.2023.6}
Ruben Turkenburg, Harsh Beohar, Clemens Kupke, and Jurriaan Rot.
\newblock {Forward and Backward Steps in a Fibration}.
\newblock In Paolo Baldan and Valeria de~Paiva, editors, {\em 10th Conference
  on Algebra and Coalgebra in Computer Science (CALCO 2023)}, volume 270 of
  {\em Leibniz International Proceedings in Informatics (LIPIcs)}, pages
  6:1--6:18, Dagstuhl, Germany, 2023. Schloss Dagstuhl -- Leibniz-Zentrum
  f{\"u}r Informatik.

\bibitem{bw:behavioural-pseudometric}
Franck {van Breugel} and James Worrell.
\newblock A behavioural pseudometric for probabilistic transition systems.
\newblock {\em Theoretical Computer Science}, 331:115--142, 2005.

\bibitem{g:linear-branching-time}
Rob van Glabbeek.
\newblock The linear time -- branching time spectrum~{I}.
\newblock In J.A. Bergstra, A.~Ponse, and S.A. Smolka, editors, {\em Handbook
  of Process Algebra}, chapter~1, pages 3--99. Elsevier, 2001.

\bibitem{v:optimal-transport}
C{\'e}dric Villani.
\newblock {\em Optimal Transport -- Old and New}, volume 338 of {\em A Series
  of Comprehensive Studies in Mathematics}.
\newblock Springer, 2009.

\bibitem{WildSchroder21}
Paul Wild and Lutz Schr{\"{o}}der.
\newblock A quantified coalgebraic van {B}enthem theorem.
\newblock In {\em Proc. of FOSSACS '21}, volume 12650 of {\em LNCS}, pages
  551--571. Springer, 2021.

\bibitem{WildSchroder22}
Paul Wild and Lutz Schr{\"{o}}der.
\newblock Characteristic logics for behavioural hemimetrics via fuzzy lax
  extensions.
\newblock {\em Logical Methods in Computer Science}, 18(2), 2022.

\bibitem{WildEA18}
Paul Wild, Lutz Schr{\"{o}}der, Dirk Pattinson, and Barbara K{\"{o}}nig.
\newblock A van {B}enthem theorem for fuzzy modal logic.
\newblock In Anuj Dawar and Erich Gr{\"{a}}del, editors, {\em Logic in Computer
  Science, {LICS} 2018}, pages 909--918. {ACM}, 2018.

\end{thebibliography}

\newpage

\appendix

\end{document}